\documentclass[prx,amsmath,amssymb, notitlepage, twocolumn,
nofootinbib,
superscriptaddress,
longbibliography
]{revtex4-1}
\usepackage{mathtools}
\usepackage{amsmath}
\usepackage{amsthm}
\usepackage{tabularx,graphicx}
\usepackage{epstopdf}
\usepackage{makecell}
\usepackage{graphicx}
\usepackage{latexsym}
\usepackage{amssymb}
\usepackage{amsmath}
\usepackage{color, colortbl}
\usepackage{psfrag}
\usepackage{bbm}
\usepackage{bm}
\usepackage{titlesec}
\usepackage{dsfont}
\usepackage{feynmp}
\usepackage{slashed}
\usepackage{multirow}
\usepackage{url}
\usepackage{xr}
\usepackage{xcite}

\makeatletter
\newcommand*{\addFileDependency}[1]{
  \typeout{(#1)}
  \@addtofilelist{#1}
  \IfFileExists{#1}{}{\typeout{No file #1.}}
}
\makeatother

\newcommand{\onenorm}[1]{\left\| #1 \right\|_1}
\newcommand{\opnorm}[1]{\left\| #1 \right\|}

\newcommand{\beq}{\begin{eqnarray}}
\newcommand{\eeq}{\end{eqnarray}}

\newcommand{\ra}{\rangle}

\newcommand{\tr}{{\rm tr}}

\newcommand{\bsp}{\begin{aligned}}
\newcommand{\esp}{\end{aligned}}

\newcommand{\hc}{{\rm h.c.}}

\newcommand{\ie}{{i.e., }}

\newcommand{\supp}{\mathrm{supp}}

\usepackage{xcolor}
\definecolor{darkblue}{rgb}{0.,0.,0.4}
\definecolor{darkred}{rgb}{0.5,0.,0.}
\definecolor{BlueViolet}{RGB}{138,43,226}
\definecolor{SkyBlue}{RGB}{30,144,255}
\definecolor{DarkGreen}{RGB}{0,100,0}

\usepackage[colorlinks=true,linkcolor=darkblue,citecolor=blue,urlcolor=darkred]{hyperref}
\usepackage[normalem]{ulem}

\newcommand{\jinmin}[1]{ { \color{red} \footnotesize (\textsf{JY}) \textsf{\textsl{#1}} }}
\newcommand{\chong}[1]{ { \color{cyan} \footnotesize (\textsf{CW}) \textsf{\textsl{#1}} }}

\usepackage{mathrsfs}
\usepackage{comment}

\newcommand{\dist}{\mathrm{dist}}

\newtheorem{corollary}{Corollary}
\newtheorem{theorem}{Theorem}
\newtheorem{lemma}{Lemma}
\newtheorem{definition}{Definition}

\newtheorem{proposition}{Proposition}

\usepackage{nomencl}
\usepackage{pst-all}
\makenomenclature
\usepackage{leftidx}

\usepackage[off]{auto-pst-pdf}
\usepackage{booktabs}
\usepackage{yfonts}
\makeatletter

\usepackage[caption=false]{subfig}

\usepackage{enumitem}

\begin{document}

\title{Order-disorder trade-off in dirty quantum systems}

\author{Jinmin Yi}
\affiliation{Perimeter Institute for Theoretical Physics, Waterloo, Ontario, Canada N2L 2Y5}
\affiliation{Department of Physics and Astronomy, University of Waterloo, Waterloo, Ontario, Canada N2L 3G1}

\author{Chong Wang}
\affiliation{Perimeter Institute for Theoretical Physics, Waterloo, Ontario, Canada N2L 2Y5}

\begin{abstract}
We prove a trade-off theorem for order and disorder parameters in one-dimensional quantum spin systems with quenched disorder. For a disordered ensemble with exact Ising symmetry and average translation symmetry, any gapped ensemble must have one and only one of the following: an $O(1)$ order parameter or an $O(1)$ disorder parameter with even parity, both of the Edwards-Anderson type. The result extends to nearly gapped ensembles that accommodate Griffiths-type rare-region effects. These results offer a powerful and rigorous framework to understand the disorder effects beyond perturbative approaches. As applications, we (1) establish the existence of string order parameters for SPT phases; (2) derive a Lieb-Schultz-Mattis-type constraint for disordered ensembles, which requires a nearly gapped ensemble to spontaneously break the symmetry; and (3) discuss similar trade-off relations for disordered fermion chains, leading to an improved understanding of certain ``intrinsically disordered'' topological phases.
\end{abstract}

\maketitle

\section{Introduction}

Quenched disorder plays a central role in condensed matter physics. In realistic dirty materials, its effects are both unavoidable and profound, ranging from disorder-induced localization to the quantized Hall effect. A key challenge in understanding disorder effects~\cite{Sachdev2011} is that disorder can give rise to intrinsically non-perturbative phenomena like rare region effects that lie beyond the reach of perturbative approaches~\cite{PhysRevLett.23.17, Vojta_2006, PhysRevB.51.6411, IGLOI_2005, PhysRevB.111.165114}. Developing a rigorous framework for phases of dirty quantum systems therefore requires tools that can handle such non-perturbative effects.

In this work, we establish a rigorous framework for understanding the effects of quenched disorder in the context of one-dimensional quantum spin chain. We focus on systems with an exact Ising symmetry and average translation symmetry, and we study the order and disorder parameters for such systems. The order and disorder parameters are powerful tools for understanding both the symmetry breaking phases and the symmetric topological phases in 1D clean systems, which also provides a useful foundation for the recent advance of the topological holography principle and symmetry topological field theory (SymTFT)~\cite{PhysRevResearch.2.033417, PhysRevB.107.155136, PhysRevResearch.2.043086, Huang_2025}.

In a clean chain with Ising symmetry $S=\prod_i S_i$, we say the system has an order parameter if there exist parity-odd local operators $O_i$ such that 
\begin{equation}
    \lim_{|i-j|\to \infty} \langle O_i O_j\rangle \neq 0,
\end{equation}
detecting spontaneous symmetry breaking. Analogously, the system has a disorder parameter if there exist $O_i$, $O_j$ with the same parity under $S$ such that the string correlator
\begin{equation}
    \lim_{|i-j|\to \infty} \langle O_i O_j \prod_{i<k<j} S_k \rangle \neq 0,
\end{equation}
detecting the symmetric phase. Levin~\cite{Levin_2020} proved rigorously that any gapped, translationally invariant Ising chain has either an order parameter or a disorder parameter with even parity, but not both.

In this work, we extend this dichotomy to quenched disordered systems. Since correlators can fluctuate in sign across disorder realizations, the correct diagnostic is the disorder average of the absolute value---the analogue of the Edwards-Anderson order parameter~\cite{SFEdwards_1975}. Thus for a disorder ensemble $\{|\Omega_\omega\rangle\}$ with Ising symmetry $S=\prod_i S_i$ we say it has an order parameter if there exist parity-odd local operators $O_i$ such that
\begin{equation}
    \lim_{|i-j|\to \infty}\mathbb{E}_\omega\!\left(\left|\langle \Omega_\omega|O_iO_j|\Omega_\omega\rangle\right|\right)\neq 0,
\end{equation}
where $\mathbb{E}_\omega$ denotes the average over the ensemble,
and analogously, it has a disorder parameter if there exist local operators $O_i$, $O_j$ such that
\begin{equation}
    \lim_{|i-j|\to \infty}\mathbb{E}_\omega\!\left(\left|\langle \Omega_\omega|O_i O_j \prod_{i<k<j} S_k|\Omega_\omega\rangle\right|\right)\neq 0.
\end{equation}
Our main result is that the clean dichotomy survives for gapped ensembles (all ground states uniformly gapped in its symmetry sector and in the same phase): any such ensemble has either an order parameter or a disorder parameter, but not both, and the endpoint operators of any disorder parameter must be even under $S$. The proof combines two ingredients: the principle that local perturbations perturb locally (LPPL)~\cite{Bachmann_2011,De_Roeck_2015}, which ensures that local order/disorder indicators depend only on nearby disorder variables, and a probabilistic analysis that relates the probabilities of local ordering and disordering events to the existence of global disorder-averaged parameters. The result extends to nearly gapped ensembles that accommodate rare Griffiths regions occurring with exponentially small probability in their size.

Physically, the necessity of taking the absolute value before disorder averaging is straightforward to understand. For the order parameter, a spin-glass phase in which $\langle O_i O_j\rangle_{\Omega_\omega}$ acquires random signs across the ensemble causes the first moment to vanish, while the Edwards–Anderson correlator remains finite. Similarly, for the disorder parameter, an ``Anderson insulator'' phase with localized yet randomly distributed $\mathbb{Z}_2$ charge produces fluctuating signs of the disorder parameter, so that only the absolute moment survives disorder averaging. In our framework, both the random magnet and the Anderson insulator are viewed as simple gapped ensembles.

Disorder also has a rich interplay with topological phases of matter. In this work we provide a rigorous handle on three  examples: 
\begin{enumerate}
    \item For symmetry-protected topological (SPT)~\cite{Chen_2012} phases in $1d$, we show that the Edwards-Anderson type disorder-averaged string order parameters~\cite{PhysRevB.45.304,P_rez_Garc_a_2008,Pollmann_2012,Haegeman_2012} remain well-behaved diagnostics for different phases. 
    \item We derive a Lieb-Schultz-Mattis-type constraint~\cite{Lieb_1961,LSM_oshikawa,LSM_HigherD_Hastings,ChengPRX2016,liu2025entanglement,yi2025lovasz,Gioia2021,ZouLSMreview} for disordered ensembles~\cite{Kimchi_2018,PhysRevX.13.031016,YouOshikawa,Xu_2025,panahi2026quantumcriticalitystrongrandomness,liu2026average}: if a disorder ensemble has an average translation symmetry and an exact $\mathbb{Z}_2\times\mathbb{Z}_2$ symmetry with a projective representation on each site, it cannot be nearly gapped without long-range order---it must develop an order parameter and break some symmetry. 
    \item Using Jordan-Wigner transform, the trade-off result can be translated to a similar result on fermion chains, now constraining two types of disorder operators, dressed by bosonic and fermionic operators, respectively. This result can be used to understand a class of ``intrinsically disordered'' topological phase proposed in Ref.~\cite{Ma_2025}.
\end{enumerate}

While many of the above results have been anticipated in some form in earlier literature based on physical arguments, our work puts them on a solid footing -- this is particularly valuable for strongly disordered and interacting systems, where non-perturbative effects, such as rare regions, often render conventional methods uncontrolled. Indeed, without a proper understanding of rare-region effects, it was difficult even to formulate the correct statements, let alone establish them rigorously. From this perspective, one important implication of our work is that the notion of a ``nearly gapped ensemble'' provides the natural generalization of the familiar concept of a ``gapped state'' in clean systems.

The rest of this paper is organized as follows. In Sec.~\ref{sec: Preliminaries} we introduce the setting, define order and disorder parameters for disordered ensembles, review the  LPPL principle and our model of rare-region effects, and provide a road map for our main results. In Sec.~\ref{sec: tradeoff_gapped} we prove the order-disorder trade-off theorem for gapped ensembles: Sec.~\ref{subsec:gapped-setup} fixes technical notation, Sec.~\ref{subsec: events} introduces the probabilistic events for order and disorder, and Secs.~\ref{subsec:tradeoff_proof}--\ref{subsec:disorderparity} establish the three parts of the main theorem (existence, mutual exclusivity, and parity of the disorder parameter). In Sec.~\ref{sec:Rareregion} we extend these results to nearly gapped ensembles with rare-region effects. In Sec.~\ref{sec: Application} we present the applications to average SPT string order (Sec.~\ref{subsec:aspt}),  the LSM-type constraint (Sec.~\ref{subsec:disorderLSM}) and fermion chains (Sec.~\ref{subsec:JordanWigner}). We conclude with a summary and outlook. Several appendixes contain technical details.

\section{Preliminaries and Summary of results}\label{sec: Preliminaries}

This section sets up the framework for our main results and provides a summary of the key statements.
In Sec.~\ref{subsec: setting} we describe the physical setting and introduce the notion of a gapped ensemble. In Sec.~\ref{subsec: parameters} we define order and disorder parameters for dirty systems. In Sec.~\ref{subsec: LPPL} we review the principle that local perturbations perturb locally (LPPL) and how we model the rare region effects. Sec.~\ref{subsec: results} summarizes the main results.
\subsection{Gapped disordered systems}\label{subsec: setting}
In this work, we focus on spin systems on a 1D lattice with periodic boundary conditions. We consider a Hilbert space with a local tensor product structure $\mathcal{H}=\otimes_i \mathcal{H}_i$, where $i$ labels lattice sites. We study quenched disorder by considering an ensemble of local Hamiltonians ${H_\omega}$ depending on the disorder realization $\omega$. Concretely, the Hamiltonian $H_\omega$ for each realization $\omega$ takes the form
\begin{equation}
H_\omega=H_0+\sum_i (v_i^\omega \mathcal{O}_i+h.c.),
\end{equation}
where $v_i^\omega$ denotes the disorder potential at site $i$ for realization $\omega$, and $\mathcal{O}_i$ is a local operator supported near $\omega$. We assume the disorder is at most short-range correlated, namely that $\overline{v_i^*v_j}$ decays exponentially with $|i-j|$. 

We focus on systems with an exact Ising symmetry
\begin{equation}
S=\prod_i S_i,\qquad S_i^2=1,
\end{equation}
so that $[H_\omega,S]=0$ for every realization $\omega$. We also assume average lattice translation symmetry: the clean Hamiltonian $H_0$ commutes with the translation operator $\hat{T}$, $[H_0,\hat{T}]=0$, and the disorder distribution is invariant under lattice translations. Thus, while a given realization $H_\omega$ need not be translationally invariant, the joint distribution of the random potentials is. A standard example is an independent, identically distributed disorder, with the marginal distribution $\mathbb{P}(v_i)$ the same for every site $i$.

We denote by $\{|\Omega_\omega\rangle\}$ the corresponding ensemble of ground states, which we take to be even under $S$ for every realization $\omega$. We call $\{|\Omega_\omega\rangle\}$ a \emph{gapped ensemble} if, for every realization $\omega$, $H_\omega$ is gapped in the even symmetry sector with a gap lower-bounded by some constant $\epsilon$, and all $|\Omega_\omega\rangle$ are smoothly connected by a quasi-adiabatic evolution, \ie they lie in the same gapped phase. This rules out mixtures of states from different phases (phase separation), which is also a natural assumption when $v_i^\omega$ is a continuous variable. Note that this notion excludes Griffiths-type rare regions that lead to small gaps; in Sec.~\ref{subsec: LPPL} we relax this to a \emph{nearly gapped ensemble}, where such rare regions are allowed but occur with exponentially small probability in their size.

The assumption that every $|\Omega_\omega\rangle$ is even under $S$ is purely technical: it ensures the ground state is uniquely determined by $\omega$ without having to track degeneracy. Even if $H_\omega$ undergoes spontaneous symmetry breaking for some $\omega$, the even- and odd-sector ground states are locally indistinguishable — no local operator can tell them apart. Fixing the even-sector ground state therefore does not affect the physical content of our results.

In our notion of a gapped ensemble, the requirement that $|\Omega_\omega\rangle$ is a gapped ground state in its own symmetry sector is weaker than the SRE-ensemble condition in the literature of average SPT~\cite{Ma_2025}, where each $|\Omega_\omega\rangle$ is required to be short-range entangled. Our setting allows, for example, an Ising ferromagnet with $|\Omega_\omega\rangle$ in a GHZ sector.

\subsection{Order and disorder parameters}\label{subsec: parameters}
We now give precise definitions of order and disorder parameters for a disordered ensemble. Intuitively, this is just the ensemble average of the absolute value of the order and disorder parameters defined for a single state as addressed in the introduction, in the spirit of the Edwards-Anderson order parameter for spin glasses~\cite{SFEdwards_1975}.  We now state the following quantitative definitions for the order parameter. Throughout, $\opnorm{\cdot}$ denotes the operator norm, \ie the largest singular value of an operator. 

\begin{definition}[Order parameter]\label{def: order-dirty}
    A disorder ensemble $\{|\Omega_\omega\ra\}$ is called to have a $(\delta,\ell)$ order parameter if, for any two sites $i$ and $j$ with $\dist(i,j)>2\ell$, there exist operators $O_{i,j}$, such that

\begin{enumerate}
    \item $O_{i,j}$ is odd under $S$.
    \item $O_{i}$ is supported on $R_{i}=[i+1, i+\ell]$ and likewise for $O_{j}$. 
    \item $ \mathbb{E}_\omega(|\langle\Omega_\omega| O_iO_j|\Omega_\omega\rangle|)\geq \delta$, where $\mathbb{E}_\omega$ is the ensemble average over disorder realizations.
    \item $\opnorm{O_{i,j}} \leq 1$.
\end{enumerate}
\end{definition}
Similarly, we can define the disorder parameter:
\begin{definition}[Disorder parameter]\label{def: disorder-dirty}
    A disorder ensemble $\{|\Omega_\omega\ra\}$ is called to have a $(\delta,\ell)$ disorder parameter if, for any two sites $i$ and $j$ with $\dist(i,j)>2\ell$, there exist operators $O_{i,j}$, such that

\begin{enumerate}
    \item $O_i$ and $O_j$ are both even under $S$ or both odd under $S$.
    \item $O_{i}$ is supported on $R_{i}=[i+1, i+\ell]$ and likewise for $O_{j}$. 
    \item $ \mathbb{E}_\omega\left(\left|\langle\Omega_\omega| O_iO_j\prod_{k=i+1}^j S_k|\Omega_\omega\rangle\right|\right)\geq \delta$, where $\mathbb{E}_\omega$ is the ensemble average over the disorder distribution.
    \item $\left\|O_{i,j}\right\| \leq 1$.
\end{enumerate}
\end{definition}
The endpoint operators $O_i$ and $O_j$ are either both even or both odd under $S$, so their parity is well-defined. We say the disorder parameter has \emph{even parity} if $O_i$ is even under $S$, and \emph{odd parity} if $O_i$ is odd under $S$. For brevity, we say a system has an $O(1)$ order (resp.\ disorder) parameter if it has a $(\delta,\ell)$ order (resp.\ disorder) parameter for some $\delta,\ell = O(1)$ independent of system size.

In the clean limit, where there is no disorder potential, these definitions reduce to the order and disorder parameters introduced in Ref.~\cite{Levin_2020}, which characterize the symmetry-broken and symmetric phases, respectively. In particular, Ref.~\cite{Levin_2020} showed that for a translational symmetric clean system that is gapped in each symmetry sector, the symmetric ground state must admit either an order parameter or a disorder parameter, but not both, and any disorder parameter must be parity even.

\subsection{LPPL and rare region effects}\label{subsec: LPPL}
One key assumption of this work is the principle of ``local perturbation perturbs locally" (LPPL)~\cite{Bachmann_2011,De_Roeck_2015}, which we now review. Precisely, suppose the disorder strengths $v_i^\omega$ are continuously tunable and the spectral gap remains uniformly bounded below throughout. If one modifies the disorder strengths only on a region $R$, passing from realization $\omega$ to $\omega'$, then the ground states before and after the change are related by a symmetric quasi-local unitary $U_R$ supported near $R$~\cite{Hastings_2005,Bravyi_2010, Bachmann_2011,Haah_2021,yi2025universaldecayconditionalmutual},
\begin{equation}
    |\Omega_{\omega'}\rangle = U_R\,|\Omega_\omega\rangle,
\end{equation}
with $|\Omega_\omega\rangle$ and $|\Omega_{\omega'}\rangle$ being the parity even ground states of $H_\omega$ and $H_{\omega'}$, respectively. 

Consequently, if one considers any local operator $O_{R}$ supported in the region $R$, its expectation value $\langle\Omega_\omega|O_R|\Omega_\omega\rangle$ should only depend on the value of $v_i^\omega$ for sites $i$ sufficiently close to the region $R$, \ie
\begin{equation}
    i \in R^{+\xi}\coloneqq \{j|\dist(j,R)\leq \xi\}\;,
\end{equation}
An immediate consequence is that the expectation value of any operator $O_R$ supported on $R$ depends only on the disorder strengths within the enlarged region $R^{+\xi} \coloneqq \{j \mid \dist(j,R)\leq \xi\}$, where $\xi$ is the correlation length: modifying $v_i^\omega$ outside $R^{+\xi}$ induces a unitary supported outside $R$, which commutes with $O_R$ and leaves $\langle\Omega_\omega|O_R|\Omega_\omega\rangle$ unchanged. Throughout, we approximate $U_R$ as strictly local on a finite neighborhood of $R$, dropping its exponentially decaying tail. This does not affect our results, and a more precise estimate of the approximation error can be found in Refs.~\cite{Haah_2021, yi2025universaldecayconditionalmutual}.

We now specify the class of rare regions we consider. We define a rare region as a spatial region $R$ on which the system fails to be gapped locally, but where the ground state outside of $R$ still resembles a gapped ground state. For simplicity, we assume that the probability of such a rare region to occur is exponentially small in its size, \ie as $p^{|R|}$ for some small $p<1$. More precisely, if $|\Omega'\rangle$ denotes a state containing a rare region supported on $R$, we model it as 
\begin{equation}
    |\Omega'\rangle=U_R|\Omega\rangle, 
\end{equation}
where $|\Omega\rangle$ is a parity-even gapped ground state and $U_R$ is a symmetric unitary supported on region $R$ and is not necessarily quasi-local within $R$. This decomposition allows rare regions to act as defects on the region $R$ while preserving the gapped structure elsewhere. The physical example we have in mind is the Griffiths phase in Ising model away from critical point, where a large region $R$ may realize a different phase than the rest of the system due to rare disorder configurations. Although the probability of the rare region effect decays with $|R|$, its effect enhances as $R$ becomes larger -- for example, the gap of the system may decay rapidly with $R$, rendering the system effectively gapless. We therefore call such ensembles \textit{nearly gapped}. More precisely,
\begin{definition}
    An ensemble of ground states is nearly gapped if it is a mixture of parity-even gapped states and rare-region states, with the probability of the rare-region states decaying exponentially in the size of the rare region.
\end{definition}

As we shall discuss later in Sec.~\ref{sec:Rareregion}, the results in this work is unaffected by such rare regions and hold generally for \textit{nearly gapped ensembles}.

The exponential suppression of the rare region probability $p^{|R|}$ is closely related to the LPPL principle: if LPPL holds, then the appearance of a large rare region $R$ requires all the random potential inside $R$ to take some rare value, which must come with probability $\sim p^{|R|}$. A class of counter examples is the infinite-randomness fixed points, including the random-singlet state~\cite{Fisher1994} and the critical random transverse-field Ising model~\cite{PhysRevB.51.6411}, in which the probability for establishing a large rare region only decays algebraically, and is therefore not considered as nearly gapped states.

\subsection{Summary of results}\label{subsec: results}
We now summarize the key results in this work. We also point out the sections that discuss these statements in detail. This section can be viewed as a road map of the work.\\

(1) A gapped ensemble of disordered systems must admit either an $O(1)$ order parameter or an $O(1)$ disorder parameter, but not both. Moreover, any disorder parameter must be parity even. 
\begin{theorem}\label{thm:tradeoff}
     A gapped ensemble $\{|\Omega_\omega\rangle\}$ has one and only one of the following: an $(\eta,\ell)$ order parameter or an $(\eta',\ell')$ disorder parameter with even parity, for large enough $O(1)$ constants $\ell$ and $\ell'$, where $\eta$ and $\eta'$ satisfy the following trade-off relation:
     \begin{eqnarray}
         c_\eta\sqrt{\eta}+c_\eta'\eta'^{1/3}>1.
     \end{eqnarray}
     where $c_\eta$ and $c_\eta'$ are $O(1)$ constants independent of the system size.
\end{theorem}

The precise requirements on these $O(1)$ constants are given in Proposition~\ref{prop:tradeoff} in Sec.~\ref{subsubsec:tradeoff-gapped}. 

The theorem has three parts: existence (the ensemble has at least one of the two: order parameter or disorder parameter, see Sec.~\ref{subsec:tradeoff_proof}), mutual exclusivity (it cannot have both, see Sec.~\ref{subsec:exclusion}), and the parity constraint (any disorder parameter must be even-parity, see Sec.~\ref{subsec:disorderparity}).

(2) Our trade-off theorem still holds with the presence of rare regions, as long as the probability of rare regions decays sufficiently fast with their size, as shown in Sec.~\ref{sec:Rareregion}. Similar to the gapped case, in Sec.~\ref{subsec: tradeoff-rare} we provide the proof of the existence of at least one of the two parameters, and in Sec.~\ref{subsec: exclusion-rare} we show that they cannot coexist. The parity constraint for disorder parameters also continues to hold in the presence of rare regions, as shown in Sec.~\ref{subsec: disorderparity-rare}.

(3) Our definition of the disorder parameters can be applied to characterize SPT phases. In particular, we demonstrate this for a disordered system in an SPT phase with exact $\mathbb{Z}_2\times \mathbb{Z}_2$ symmetry and average translational symmetry. Precisely, we have the following corollary:
\begin{corollary}\label{cor}
    Let $\{H_\omega\}$ be a disordered average-SPT ensemble with average translational
symmetry and exact $\mathbb{Z}_2\times \mathbb{Z}_2$ symmetry $S^e$ and $S^o$, with
\begin{equation}
    S^e = \prod_{i} Z_{2i}, \quad S^o = \prod_{i} Z_{2i+1}.
\end{equation}
Then, for any sites $i<j$, there exist local operators
$O_i$ and $O_j$, supported near $i$ and $j$ respectively, such that
\begin{equation}
    \mathbb{E}_\omega\left(\left|\bigl\langle O_i O_j \prod_{i<2k<j} Z_{2k} \bigr\rangle_{\Omega}\right|\right) = O(1).
\end{equation}
Similarly, there exist local operators $O_i'$ and $O_j'$, supported near
$i$ and $j$ respectively, such that
\begin{equation}
    \mathbb{E}_\omega\left(\left|\bigl\langle O_i' O_j' \prod_{i<2k+1<j} Z_{2k+1} \bigr\rangle_{\Omega}\right| \right)= O(1).
\end{equation}
Moreover, for any fixed ensemble $\{|\Omega_\omega\rangle\}$, the symmetry parities of these disorder parameters are uniquely determined. More precisely,
the operators $O_i$ and $O_j$ must be even under $S^e$, and the operators
$O_i'$ and $O_j'$ must be even under $S^o$. In addition, the parities of
$O_i, O_j$ under $S^o$ and of $O_i', O_j'$ under $S^e$ are the same, and this parity is uniquely fixed by the ensemble.
\end{corollary}

This provides a concrete application of our main theorem and suggests that string order parameters remain a powerful diagnostic tool for characterizing SPT phases even in dirty quantum systems. See Sec.~\ref{subsec:aspt} for the proof and further discussion.

(4) Our framework also allows us to derive a Lieb-Schultz-Mattis-type constraint for disordered ensembles. Precisely, we have the following corollary:
\begin{corollary}\label{cor:LSM}
    Let $\{H_\omega\}$ be a nearly gapped disordered ensemble with average translational symmetry and an exact $\mathbb{Z}_2\times \mathbb{Z}_2$ symmetry $S^x$ and $S^z$ on a periodic chain of even length, with
\begin{equation}
    S^x = \prod_{i} X_i, \quad S^z = \prod_{i} Z_i,
\end{equation}

     Then $\{H_\omega\}$ must have an $O(1)$ order parameter for either $S^x$ or $S^z$.
\end{corollary}
This result can be viewed as a disordered analogue of the Lieb-Schultz-Mattis (LSM) theorem, in the spirit that under certain symmetry constraints, the system cannot be gapped and symmetric without long-range order. In this case, the presence of an average translation symmetry and a projective representation of the $\mathbb{Z}_2\times \mathbb{Z}_2$ symmetry on each site forces the system to develop an order parameter, ruling out the possibility of a symmetric, gapped phase without long-range order. See Sec.~\ref{subsec:disorderLSM} for the proof and further discussion.

(5)Through a Jordan-Wigner transform, we can translate our results on Ising spin chains to fermion chains. Specifically, we have
\begin{corollary}
    For a fermion chain in a nearly gapped ground state ensemble, there is always some disorder parameter
    \begin{equation}
        \lim_{|i_1-i_2|\to\infty}\mathbb{E}_{\omega}|\langle O_1 O_2\prod_{i_1\leq k\leq i_2}F_k\rangle_{\Omega_\omega}|= O(1),
    \end{equation}
    where $F_k$ is the local fermion parity on site $k$. For a given nearly gapped ensemble, the dressing local operators $O_{1,2}$ are either always bosonic (a trivial phase) or always fermionic (a Kitaev chain). 
\end{corollary}
This shows that string operators remain well-behaved diagnostics for $1d$ fermionic topological phases. 

The above result can be used to understand a class of ``intrinsically disordered'' topological phase proposed in Ref.~\cite{Ma_2025}. The simplest example is a Kitaev chain, but with an additional average $U(1)$ symmetry $c_i\to e^{i\alpha}c_i$. This average symmetry would forbid a first-moment fermionic disorder parameter: $\lim_{|i_1-i_2|\to\infty}\mathbb{E}_{\omega}\langle O_1 O_2\prod_{i_1\leq k\leq i_2}F_k\rangle_{\Omega_\omega}=0$, but does not forbid the Edwards-Anderson type of disorder parameter. This is why such phase can only appear in disordered ensembles, where the $U(1)$ symmetry only holds on average and first-moment and Edwards-Anderson disorder parameters can behave very differently.

For more details on fermion chains see Sec.~\ref{subsec:JordanWigner}.

\section{Order-disorder trade-off for gapped ensembles}\label{sec: tradeoff_gapped}
In this section we prove our main theorem constraining the order and disorder parameters for gapped ensemble without rare region effects. The argument proceeds in three steps. First, we show that the system must exhibit either an order parameter or a disorder parameter. Second, we show that these two possibilities are mutually exclusive. Third, we prove that any disorder parameter must be parity even.

The proof is probabilistic in nature. We regard the disorder potentials $\{v_i^\omega\}$ as random variables, and local observables evaluated in a fixed ground state $|\Omega_\omega\rangle$ as events that depend only on the disorder variables near their support.

In Sec.~\ref{subsec:gapped-setup}, we establish the technical setup, introducing notions of strong and weak order and disorder on disjoint intervals and reviewing the foundational lemmas. In Sec.~\ref{subsec: events}, we define two families of probabilistic events $\{A_i\}$ and $\{B_i\}$, corresponding to the ground state $|\Omega_\omega\rangle$ being ordered or disordered near site $i$. In Sec.~\ref{subsec: incompatibility}, we show that the complementary events $\overline{A_i}$ and $\overline{B_i}$ are mutually incompatible, \ie at least one of $A_i$ and $B_i$ must occur for each $i$. In Sec.~\ref{subsec:tradeoff_proof}, we use these events to prove that the system must have either an $(\eta,\ell)$ order parameter or an $(\eta',\ell')$ disorder parameter, with $\eta,\eta',\ell,\ell'$ being $O(1)$ constants independent of system size. In Sec.~\ref{subsec:exclusion}, we show that the system cannot simultaneously possess both an order and a disorder parameter. In Sec.~\ref{subsec:disorderparity}, we show that any disorder parameter must be parity even.

\subsection{Technical setup for the proof}\label{subsec:gapped-setup}
We begin by establishing key definitions and notations required for the proof. Following Ref.~\cite{Levin_2020}, we define notions of strong and weak order and disorder on disjoint intervals for an individual state. These concepts are essential for characterizing the occurrence of order and disorder parameters. We also review several foundational lemmas that underpin our main arguments.

For any subset $X\subset \{1, \ldots, L\}$, define the restricted symmetry operator
\begin{equation}
    S_X \coloneqq \prod_{i\in X} S_i\;.
\end{equation}

We now introduce the notion of strong or weak ordering for an individual state $|\psi\rangle$:
\begin{definition}
A state $|\psi\rangle$ is $\delta$ strongly-ordered on disjoint intervals $I_1, I_2$ if there exist operators $O_1$ and $O_2$ such that:
\begin{enumerate}
    \item  $O_1$ and $O_2$ are supported on $I_1$ and $I_2$.
    \item  $O_1$ and $O_2$ are odd under $S$.
    \item  $\left|\left\langle O_1 O_2\right\rangle_\psi\right| \geq \delta .$
    \item  $\left\|O_1\right\|,\left\|O_2\right\| \leq 1$.
\end{enumerate}
\end{definition}
\begin{definition}
    $A$ state $|\psi\rangle$ is $\delta$ weakly-ordered on disjoint intervals $I_1, I_2$ if there exists an operator A such that:
    \begin{enumerate}
        \item $A$ is supported on $I_1 \cup I_2$.
    \item  $A$ is odd under $S_{I_1}$ and $S_{I_2}$.
    \item  $\left|\langle A\rangle_\psi\right| \geq \delta$.
    \item  $\|A\| \leq 1$.
    \end{enumerate}
\end{definition} 
Note that the only difference between the strong or weak ordering is that for the weak ordering, we do not require a tensor product structure of the operator $A$.

The notion of strong or weak disorder is defined similarly:
\begin{definition}
    A state $|\psi\rangle$ is $\delta$ strongly-disordered on disjoint, non-adjacent intervals $I_1, I_2$ if there exist operators $O_1$ and $O_2$ such that
    \begin{enumerate}
        \item $O_1$ and $O_2$ are supported on $I_1$ and $I_2$.
        \item $O_1$ and $O_2$ are both even or both odd under $S$.
    \item $\left|\left\langle O_1 O_2 S_J\right\rangle_\psi\right| \geq \delta$ where $J$ is the interval between $I_1$ and $I_2$, going clockwise.
    \item $\left\|O_1\right\|,\left\|O_2\right\| \leq 1$.
    \end{enumerate}

\end{definition}

Notice that $O_1, O_2$ are either both even or both odd under $S$: in the former case, we will say that $|\psi\rangle$ is $\delta$ strongly-disordered with even parity, while in the latter case we will say that $|\psi\rangle$ is $\delta$ strongly-disordered with odd parity.

\begin{definition}
    A state $|\psi\rangle$ is $\delta$ weakly-disordered on disjoint, non-adjacent intervals $I_1, I_2$ if there exists an operator $A$ such that
    \begin{enumerate}
        \item $A$ is supported on $I_1 \cup I_2$.
        \item $A$ is even under both $S_{I_1}, S_{I_2}$ or odd under both $S_{I_1}, S_{I_2}$.
    \item $\left|\left\langle A S_J\right\rangle_\psi\right| \geq \delta$ where $J$ is the interval between $I_1$ and $I_2$, going clockwise.
    \item $\|A\| \leq 1$.
    \end{enumerate}
\end{definition} 

Similarly to above, we will say that $|\psi\rangle$ is $\delta$ weakly-disordered with even or odd parity depending on whether $A$ is even or odd under $S_{I_1}, S_{I_2}$.

\begin{figure}
    \centering
    \includegraphics[width=0.6\linewidth]{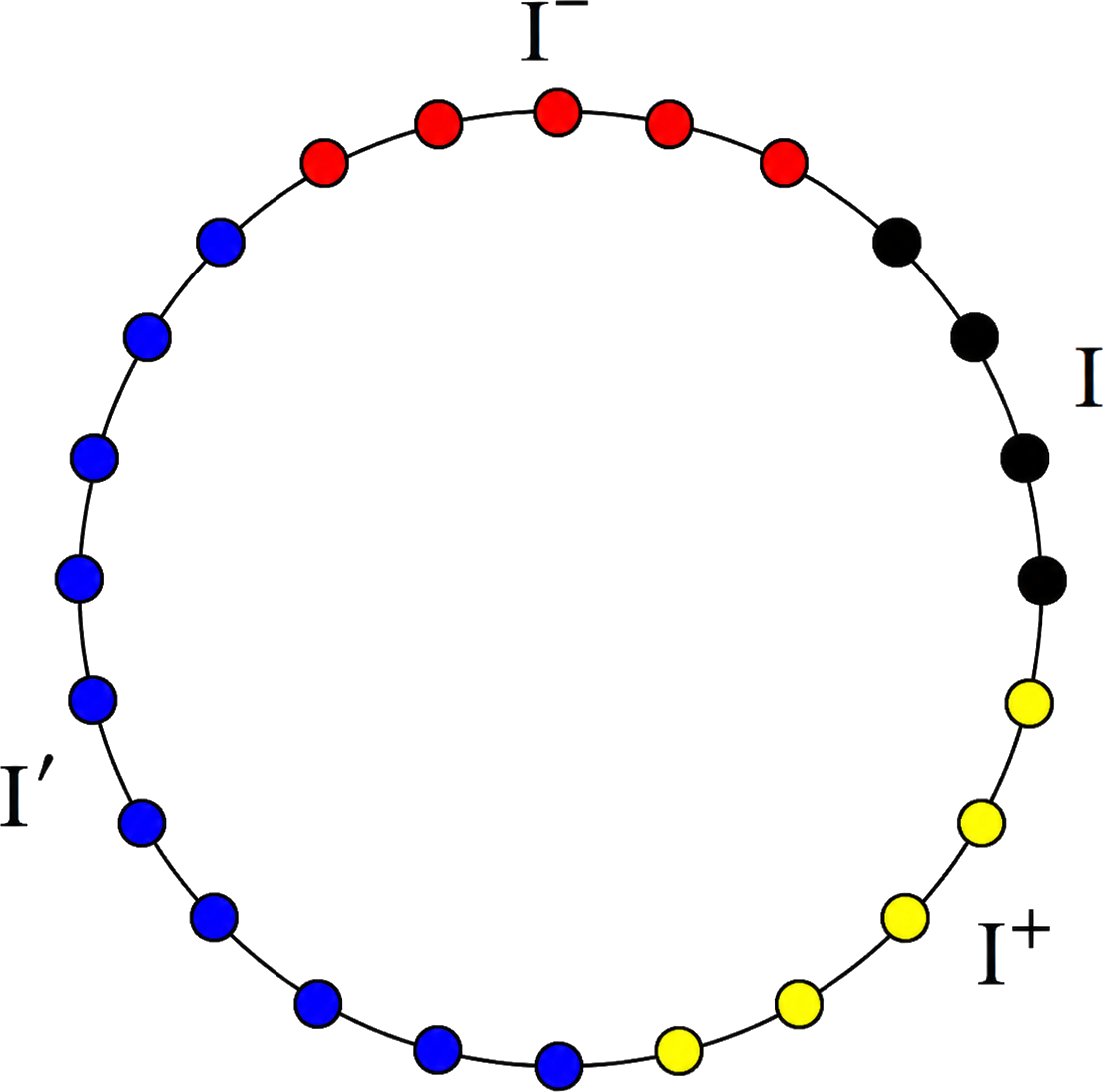}
    \caption{Adapted from Ref.~\cite{Levin_2020}. Example of $I^{ \pm}$ notation: the spins in the interval $I$ are shown in black; the spins in $I^{-}$ and $I^{+}$ are shown in red and yellow; and the spins in $I^{\prime} \equiv\left(I^{-} \cup I \cup I^{+}\right)^c$ are shown in blue. The region $I'$ is buffered from $I$ by $I^{\pm}$.}
    \label{fig:IpmNotation}
\end{figure}

We will also need the following notation for intervals. Given an interval $I \subset\{1, \ldots, L\}$ and an interval $I'$ not adjacent to $I$, let $I^+$ and $I^-$ denote the two intervals between $I$ and $I'$ in the clockwise and counterclockwise directions, respectively (see Fig.~\ref{fig:IpmNotation}). In general, we take $I^+$ and $I^-$ to be of the same size and larger than $I$.

Fix an integer $\ell>0$ and $\ell'=O(\ell)$, both of the same order as the correlation length and to be chosen later, and for each $i \in\{1, \ldots, L\}$, define
\begin{align}
    I_i&\coloneqq[i+1, i+\ell]\\
    I_i^- &\coloneqq [i-\ell'+1, i]\\
    I_i^+ &\coloneqq [i+\ell+1, i+\ell+\ell']\\
    I_i'&\coloneqq (I^-_i\cup I_i\cup I_i^+)^c
\end{align}
where we take the periodic boundary condition. 

It is worth highlighting that for a gapped ground state, the weak (dis)order can be upgraded to a strong (dis)order, as demonstrated by the following lemma:
\begin{lemma}[Lemma 2, Ref.~\cite{Levin_2020}]\label{lem:Levin2}
    For every $\delta>0$, there exists a length $\lambda_\delta=\tilde{O} \left(1/\varepsilon^{2}\right)$ such that for an $S$-symmetric ground state $|\Omega\rangle$ with gap $\epsilon$ in its symmetry sector, if it is $\delta$ weakly-ordered on $I_1,I_2$, then it is
$\delta/2$ strongly-ordered on $I_1,I_2$ for all $I_1,I_2$ that are separated by
a distance of at least $\lambda_\delta$. Likewise, if $|\Omega\rangle$ is $\delta$
weakly-disordered on $I_1,I_2$ with even (odd) parity, then it is $\delta/2$
strongly-disordered on $I_1,I_2$ with even (odd) parity for all $I_1,I_2$
separated by at least $\lambda_\delta$.
\end{lemma}
We remark that the weak-to-strong upgrade for the order parameters is included here for completeness only, following the framework of Ref.~\cite{Levin_2020}. The rest of this work only relies on the fact that a weak disorder parameter can be upgraded to a strong disorder parameter.

\subsection{Probabilistic events for order and disorder}\label{subsec: events}

Consider the disorder potential $\{v_i^\omega\}$ as a set of random variables which uniquely determines the state $|\Omega_\omega\rangle$ in the disorder ensemble, we define the following two sets of probabilistic events $A_i$ and $B_i$ that are conditions on local observables for the state $|\Omega_\omega\rangle$.

We first define the family of events associated with ordering. Let $A_i$ be the event that, for a disorder realization $\omega$, the corresponding state $|\Omega_\omega\rangle$ has mutual information between  $I_i$ and $I_i'$ bounded from below:
\begin{equation}\label{Eq:EventA_MI_Lowerbound}
\begin{split}
I_\omega(I_i:I_i')
&=S_{\omega}(I_i)+S_{\omega}(I_i^-I_iI_i^+)-S_{\omega}(I_i^-I_i^+)\\
&>\frac{1}{32}\;,
\end{split}
\end{equation}
where $S_\omega(I)$ denotes the von Neumann entropy of the state $|\Omega_\omega\rangle$ restricted to the region $I$. The specific choice of the lower bound $1/32$ is made for convenience in the later analysis.

Next, we define the family of events associated with disordering. Let $B_i$ be the event that
\begin{equation}
    \max_{\supp(U)\subset I_i^-\cup I_i^+}
    \left|\langle \Omega_\omega| U S_{I_i} |\Omega_\omega\rangle\right|
    \ge \frac{1}{2}.
\end{equation}
Again, the choice of threshold $1/2$ is for later convenience.

Both $A_i$ and $B_i$ depend on $|\Omega_\omega\rangle$ only through its reduced density matrix on the region $I_i^- \cup I_i \cup I_i^+$. By the LPPL principle, this reduced density matrix depends only on the disorder variables $v_j^\omega$ with $j$ close to that region, namely within the distance $O(\ell+\ell'+\xi)=O(\ell)$ from the site $i$. Therefore, if $\dist(i,j)\gg \ell$, the events $A_i$ and $A_j$ depend on disjoint sets of disorder variables and are independent; the same holds for $B_i$ and $B_j$.

In the following, we discuss in detail the physical meaning of events $A_i$ and $B_i$, and how they relate to the ordering and disordering of the state $|\Omega_\omega\rangle$ near the site $i$.

\subsubsection{The events for ordering: $\{A_i\}$}\label{subsubsec: event A}
We first summarize the physical content of $A_i$ before giving the rigorous proof. Informally, if $A_i$ occurs, the state has nontrivial long-range correlations between $I_i$ and its complement $I_i'$, signaling local ordering near site $i$. Moreover, if both $A_i$ and $A_j$ occur for two well-separated points $i$ and $j$, the local ordering near $i$ and $j$ induces an order parameter on $I_i$ and $I_j$.

Suppose $A_i$ occurs for a state $|\Omega\rangle$; for simplicity, we suppress the disorder label $\omega$. Then
\begin{equation}
I_\Omega(I_i:I_i')=S(I_i|I_i')_{\Omega_{I_i}\otimes\Omega_{I_i'}}-S(I_i|I'_i)_{\Omega}>\frac{1}{32}
\end{equation}
where the conditional entropy of a bipartite state $\sigma_{AB}$ is defined by
\begin{equation}
    S(A|B)_\sigma \coloneqq S_\sigma(AB)-S_\sigma(B),
\end{equation}
with $S_\sigma(X)$ the von Neumann entropy of the reduced density matrix $\sigma_X$ on subsystem $X$. 

The conditional entropy is continuous in the density matrix~\cite{Alicki_2004, Winter_2016}, as quantified by the following Fannes-type inequality:
\begin{lemma}[Lemma 2, Ref.~\cite{Winter_2016}]\label{lem:Fannes_Cond_entropy}
    For states $\rho$ and $\sigma$ on a Hilbert space $A \otimes B$ if $\|\rho-\sigma\|_1 \leq \epsilon \leq 2$, then
\begin{equation}
    \left|S(A|B)_\rho-S(A |B)_\sigma\right| \leq  \epsilon \log |A|+2 H_2\left(\epsilon/2\right)
\end{equation}
where $H_2(x)=-x\log x-(1-x)\log(1-x)$, and the trace-norm  $\onenorm{A}\coloneqq \tr\sqrt{ A^\dagger A}$.
\end{lemma}

Since mutual information is nonnegative, Lemma.~\ref{lem:Fannes_Cond_entropy} implies:
\begin{equation}\label{Eq:EventA_MI_UpperfromCorr}
    I_\Omega(I_i:I_i')\leq \varepsilon \ell+2H_2(\varepsilon/2)
\end{equation}
where $\varepsilon=\onenorm{\Omega_{I_i}\otimes\Omega_{I_i'}-\Omega_{I_iI_i'}}$. Combining (\ref{Eq:EventA_MI_Lowerbound}) and (\ref{Eq:EventA_MI_UpperfromCorr}), we have
\begin{equation}
    \varepsilon \ell+2H_2(\varepsilon/2)>\frac{1}{32}\;,
\end{equation}
where
\begin{equation}
    \varepsilon \coloneqq
    \left\|\Omega_{I_i}\otimes \Omega_{I_i'} - \Omega_{I_i I_i'}\right\|_1,
\end{equation}
and we denote the solution to this inequality as $\varepsilon>\varepsilon_\ell$.

The trace norm $\varepsilon$ can in turn be bounded by connected correlation functions.

\begin{lemma}[Lemma 20 of Ref.~\cite{Brand_o_2014}]\label{lem:brandao14}
For every bounded operator $L \in \mathbb{B}(\mathbb{C}^d\otimes\mathbb{C}^D)$ with $d\le D$,
\begin{equation}
    \|L\|_1 \le d^2 \max_{\|X\|\le 1,\|Y\|\le 1}
    \left|\tr\bigl((X\otimes Y)L\bigr)\right|.
\end{equation}
\end{lemma}

Applying this with
\begin{equation}
    L = \Omega_{I_i}\otimes \Omega_{I_i'} - \Omega_{I_i I_i'},
\end{equation}
we find that there exist unit-norm operators $X$ and $Y$ supported on $I_i$ and $I_i'$, respectively, such that
\begin{equation}
\begin{split}
     &|\langle\Omega|XY|\Omega\rangle-\langle\Omega|X|\Omega\rangle\langle\Omega|Y|\Omega\rangle|\\
     \geq &2^{-2\ell} \onenorm{\Omega_{I_i}\otimes\Omega_{I_i'}-\Omega_{I_iI_i'}}> 2^{-2\ell}\varepsilon_\ell
\end{split}
\end{equation}
We now decompose $X$ and $Y$ into components with definite parity under the Ising symmetry:
\begin{align}
    X_\pm &\coloneqq \frac{1}{2}(X \pm SXS),\\
    Y_\pm &\coloneqq \frac{1}{2}(Y \pm SYS).
\end{align}
Since $|\Omega\rangle$ has a definte parity, the expectation value of any odd operator vanishes. Therefore,
\begin{equation}\label{eq:even_odd_decomp}
\begin{split}
    &\left|
        \langle \Omega|XY|\Omega\rangle
        - \langle \Omega|X|\Omega\rangle \langle \Omega|Y|\Omega\rangle
    \right| \\
    &= \Bigl|
        \langle \Omega|X_-Y_-|\Omega\rangle
        + \langle \Omega|X_+Y_+|\Omega\rangle
        - \langle \Omega|X_+|\Omega\rangle \langle \Omega|Y_+|\Omega\rangle
    \Bigr| \\
    &\ge
    |\langle \Omega|X_-Y_-|\Omega\rangle|
    -
    \left|
        \langle \Omega|X_+Y_+|\Omega\rangle
        - \langle \Omega|X_+|\Omega\rangle \langle \Omega|Y_+|\Omega\rangle
    \right|.
\end{split}
\end{equation}
The first term is of the form expected for an order parameter, while the second is a connected correlator of parity-even operators and should decay exponentially with separation. 

To bound the second term, we use the cluster decomposition theorem for even operators.

\begin{lemma}[Proposition 2 of Ref.~\cite{Levin_2020}]\label{lem:factorization}
Let $O_1$ and $O_2$ be operators that are even under $S$, have norm at most $1$, and are supported on intervals $I_1$ and $I_2$ separated by distance at least $\ell'$. Then
\begin{equation}
    \left|
        \langle O_1O_2\rangle_\Omega
        - \langle O_1\rangle_\Omega\langle O_2\rangle_\Omega
    \right|
    \le
    (6+\sqrt{2})\sqrt{g(\ell')} + g(\ell'),
\end{equation}
where
\begin{equation}
    g(\ell')=\mathrm{poly}(\ell',\epsilon^{-1})e^{-c_1\epsilon \ell'}
\end{equation}
for some numerical constant $c_1>0$, and $\epsilon$ is the energy gap within the even subspace.
\end{lemma}

Since $\supp(X_+)\subset I_i$ and $\supp(Y_+)\subset I_i'$, the separation between their supports is at least $\ell'$. Because $g(\ell')$ decays exponentially, for any fixed $\ell$ one can choose $\ell'=O(\ell)$ such that
\begin{equation}\label{eq:Event_A_cond}
    (6+\sqrt{2})\sqrt{g(\ell')} + g(\ell')
    \le 2^{-2\ell-1}\varepsilon_\ell.
\end{equation}
With this choice of $\ell'$, Eq.~\eqref{eq:even_odd_decomp} implies that there exist odd operators $X_-$ and $Y_-$ such that
\begin{equation}
\begin{split}
    |\langle \Omega|X_-Y_-|\Omega\rangle|
    &\ge
    2^{-2\ell}\varepsilon_\ell
    - \left((6+\sqrt{2})\sqrt{g(\ell')} + g(\ell')\right) \\
    &\ge 2^{-2\ell-1}\varepsilon_\ell.
\end{split}
\end{equation}
Thus the event $A_i$ implies a nontrivial odd correlator between $I_i$ and $I_i'$. Since both $\ell$ and $\varepsilon_\ell$ are $O(1)$, this is an $O(1)$ lower bound. However, $I_i'$ has support comparable to the system size, so this correlator is not yet an order parameter in the sense of Definition.~\ref{def: order-dirty}.

To convert it into a genuine order parameter for the state $|\Omega\ra$, we use the following lemma.
\begin{lemma}[Extension of Lemma 3, Ref.~\cite{Levin_2020}]\label{lem:Levin3}
Let $O_1,O_1',O_2,O_2'$ be operators supported on intervals $I_1,I_1',I_2,I_2'$ respectively, all either odd or even under $S$, with norm at most $1$. Define
\begin{equation}
    \ell_{\min}=\min\bigl(\dist(I_1,I_1'),\dist(I_1',I_2'),\dist(I_1,I_2)\bigr).
\end{equation}
Then
\begin{equation}
    \left|\langle O_1 O_2\rangle_\Omega\right| \ge \left|\langle O_1 O_1'\rangle_\Omega\right| \left|\langle O_2 O_2'\rangle_\Omega\right| - f(\ell_{\min}),
\end{equation}
and
\begin{equation}
    \begin{split}
    &\left|\left\langle O_1 O_2 S_{J\cup I_2}\right\rangle_\Omega\right|\\
     \ge &\left|\left\langle O_1 O_1' S_{I_1^+}\right\rangle_\Omega\right| \left|\left\langle O_2 O_2' S_{I_2^+}\right\rangle_\Omega\right| - f(\ell_{\min}),
    \end{split}
\end{equation}
where $f(\ell)=\mathrm{poly}(\ell,\epsilon^{-1})e^{-c\epsilon \ell}$, $J$ is the interval between $I_1$ and $I_2$ in the clockwise direction, and $I_{1(2)}^+$ is the interval between $I_{1(2)}$ and $I_{1(2)}'$ in the clockwise direction. 
\end{lemma}
The proof of this extension is essentially identical to that of Lemma 3 in Ref.~\cite{Levin_2020}, with the generalization that we do not impose constraints on the separation between regions $I_i$ and $I_i'$. We therefore omit the details here.

Now suppose two events $A_i$ and $A_j$ occur with $\dist(i,j)\ge \ell'$. We choose $O_i,O_i'$ to be the odd operators $X_-,Y_-$ obtained above for site $i$, and similarly define $O_j,O_j'$ for site $j$. If $\ell'$ is chosen large enough that
\begin{equation}
    f(\ell') \le \frac{(2^{-2\ell-1}\varepsilon_\ell)^2}{2},
\end{equation}
which is possible if one takes $\ell'/\ell=O(1)$ sufficiently large since $f(\ell')$ decays exponentially, then Lemma~\ref{lem:Levin3} yields odd operators $O_i$ and $O_j$ such that
\begin{equation}
    \left|\langle O_i O_j\rangle_\Omega\right|
    \ge 2^{-4\ell-3}\varepsilon_\ell^2.
\end{equation}
Therefore, whenever two well-separated events $A_i$ and $A_j$ occur, the state exhibits an order parameter between the corresponding regions.

\subsubsection{The events for disordering: $\{B_i\}$}

We first summarize the physical content of $B_i$ before giving the rigorous argument. Informally, if $B_i$ occurs, the state has a nontrivial disorder parameter on $I_i^-$ and $I_i^+$, signaling that the system is disordered near site $i$. Moreover, if $B_{i,j,k}$ all occur for three well-separated sites $i$, $j$, $k$, the local disordering near those sites induces a disorder parameter on at least two of the three regions $I_{i,j,k}^-$.

To see that the system is disordered near site $i$ when $B_i$ occurs, recall that if $B_i$ occurs, then
\begin{equation}
    \max_{\supp(U)\subset I_i^-\cup I_i^+} |\langle \Omega_\omega|US_{I_i}|\Omega_\omega\rangle|\geq \frac{1}{2}\;.
\end{equation}
To single out the disorder parameters with definite parity, we can decompose $U$ into components with even and odd parity on $I_i^\pm$, \ie
\begin{equation}
    U_{\pm} \equiv \frac{1}{4}\left(U_*+S U_* S\right) \pm \frac{1}{4} S_{I_i^+}\left(U_*+S U_* S\right) S_{I_i^+}\;,
\end{equation}
where $U_*$ is the maximal choice of $U$. By assumption,
$\left|\left\langle\left(U_{+}+U_{-}\right) S_{I_i}\right\rangle_{\Omega_\omega}\right| \geq 1/2$. Hence either
$\left|\left\langle U_{+} S_{I_i}\right\rangle_{\Omega_\omega}\right| \geq 1/4$ or
$\left|\left\langle U_{-} S_{I_i}\right\rangle_{\Omega_\omega}\right| \geq 1/4$.
Let $U_a$ denote whichever of $U_+$ or $U_-$ gives the larger expectation value.
It is straightforward to check that $U_a$ is supported on $I_i^-\cup I_i^+$,
has definite even/odd parity under $S_{I_i^-}$ and $S_{I_i^+}$, and satisfies
$\|U_a\|\leq 1$. Thus $|\Omega_\omega\rangle$ is 1/4 weakly-disordered on $I_i^-, I_i^+$. By Lemma~\ref{lem:Levin2}, if $\ell$ is chosen large enough that $\ell\geq \lambda_{1/4}$, then $|\Omega_\omega\rangle$ is $1/8$ strongly-disordered on $I_i^-, I_i^+$ with the same parity as $U_a$.

Similarly to the case for $\{A_i\}$, we can upgrade this local disordering to a non-local disordering when multiple events occur at the same time. Suppose three events $B_{i,j,k}$ occur with the minimum distance between any two events larger than $\ell$, by the pigeonhole principle, there must be two events, say $B_i$ and $B_j$, such that $O_i^\pm$ and $O_j^\pm$ share the same parity. We can then apply Lemma.~\ref{lem:Levin3} to obtain a disorder parameter between $I_i^-$ and $I_j^-$. Concretely, if $\ell$ is chosen large enough that
\begin{equation}
    f(\ell) \le \frac{1}{128},
\end{equation}
then Lemma~\ref{lem:Levin3} implies for operators $O_i^-$ and $O_j^-$,
\begin{equation}
\begin{split}
    &\left|\langle O_i^- O_j^- S_{J\cup I_j^-}\rangle_\Omega\right| \\
    \ge&
    \left|\langle O_i^- O_i^+ S_{I_i}\rangle_\Omega\right|
    \left|\langle O_j^- O_j^+ S_{I_j}\rangle_\Omega\right|
    - f(\ell)\ge\frac{1}{128}.
\end{split}
\end{equation}
Thus the occurrence of three well-separated events $B_i$, $B_j$ and $B_k$ implies the existence of a disorder parameter between at least two of the three regions $I_i^-, I_j^-, I_k^-$.

\subsection{Incompatibility of the events $\overline{A_i}$ and $\overline{B_i}$}\label{subsec: incompatibility}
We now show that the events $\overline{A_i}$ and $\overline{B_i}$ are not compatible with each other, \ie at least one of $A_i$ and $B_i$ must happen. The proof is identical to that in Ref.~\cite{Levin_2020}, and we include it here for completeness.

We prove by contradiction. Assume that both $\overline{A_i}$ and $\overline{B_i}$ occur for a specific state $|\Omega\rangle$. Consider another state $|\Omega'\rangle=S_{I_i}|\Omega\rangle$ and consider their reduced density matrix on $X=I_i\cup I_i'$. By the Fuchs-van de Graaf inequality~\cite{fuchs1998cryptographicdistinguishabilitymeasuresquantum}
\begin{equation}\label{eq:FuchsGraaf}
    \frac{1}{2}\onenorm{\Omega_X-\Omega_X'}+F(\Omega_X,\Omega_X')\geq 1\;,
\end{equation}
where the fidelity $F(\rho,\sigma)\coloneqq\left(\tr\sqrt{\sqrt{\rho}\sigma\sqrt{\rho}}\right)^2$.

For the first term, note that
\begin{equation}
\begin{split}
    \onenorm{\Omega_X-\Omega_X'}&=\max_{\supp(A)\subset X} |\langle\Omega|A|\Omega\rangle-\langle\Omega'|A|\Omega'\rangle|\\
    &=\max_{\supp(A)\subset X} |\langle\Omega|A-S_{I_i}AS_{I_i}|\Omega\rangle|
\end{split}
\end{equation}
and the second term can be related to event $B_i$ by considering the Uhlmann's theorem~\cite{UHLMANN1976273}, \ie
\begin{equation}
\begin{split}
    F(\Omega_X,\Omega_X')&=\max_{\supp(U)\subset I_i^-\cup I_i^+}|\langle\Omega|U|\Omega'\rangle\\
    &=\max_{\supp(U)\subset I_i^-\cup I_i^+} |\langle \Omega|US_{I_i}|\Omega\rangle|\;.
\end{split}
\end{equation}

Suppose that event $B_i$ does not occur, \ie
\begin{equation}
    \max_{\supp(U)\subset I_i^-\cup I_i^+} |\langle \Omega|US_{I_i}|\Omega\rangle|<\frac{1}{2}\;
\end{equation}
by \eqref{eq:FuchsGraaf}, we have~\cite{Levin_2020}:
\begin{equation}
    \max _{\operatorname{supp}(A) \subset I_i \cup I_i'}\left|\left\langle\frac{1}{2}\left(A-S_{I_i} A S_{I_i}\right)\right\rangle_\Omega\right|>\frac{1}{2}
\end{equation}
namely, $|\Omega\rangle$ is $1/2$ weakly-ordered on $I_i$, $I_i'$. As demonstrated in Lemma.~\ref{lem:Levin2}, if we take $\ell'\geq \lambda_{1/4}$, then $|\Omega\rangle$ is $1/4$ strongly-ordered on $I_i$, $I_i'$, \ie there exist odd operators $O_i$ (resp. $O_i'$), supported on $I_i$ (resp. $I_i'$) such that 
\begin{equation}
    \left|\left\langle O_i O_i'\right\rangle_\Omega\right| >\frac{1}{4} .
\end{equation}
Note that the mutual information is lower bounded by the correlation function~\cite{2008PhRvL.100g0502W}, thus
\begin{equation}
    I_\Omega(I_i:I_i')\geq\frac{1}{2}\left|\left\langle O_i O_i'\right\rangle_\Omega\right|^2>\frac{1}{32}
\end{equation}
However, if $A_i$ does not occur, then $I_\Omega(I_i:I_i')\leq\frac{1}{32}$, a contradiction. Therefore, $\overline{A_i}$ and $\overline{B_i}$ cannot occur simultaneously, which completes the proof.

\subsection{Trade-off between order and disorder}\label{subsec:tradeoff_proof}
We are now ready to show that for any gapped ensemble $\{|\Omega_\omega\rangle\}$, there must be either an $O(1)$ order parameter or a disorder parameter. Concretely, we will assume that the system has neither a $(\eta,\ell)$ order parameter nor a $(\eta',\ell')$ disorder parameter, and obtain a trade-off between $\eta$ and $\eta'$. By choosing $\ell$ and $\ell'$ appropriately, we can show that at least one of $\eta$ and $\eta'$ must be $O(1)$.

The key observation is that if both $\eta$ and $\eta'$ are small, then the system is neither ordered nor disordered, and hence the probability of the events $A_i$ and $B_i$ must be small. By the union bound, the probability that least one event $A_i$ or $B_i$ occurs is also small. However, we have shown above that at least one of $A_i$ and $B_i$ must occur for each $i$, so
\begin{equation}
    \mathbb{P}(A_i)+\mathbb{P}(B_i)\ge 1\;,
\end{equation}
leading to a contradiction.

In the following, we estimate $\mathbb{P}(A_i)$ and $\mathbb{P}(B_i)$ in terms of $\eta$ and $\eta'$. The estimation is based on the LPPL principle, which implies that different events $A_i$ (or $B_i$) are independent if they are sufficiently separated. We also make use of the translational invariance of the disorder ensemble, so $\mathbb{P}(A_i)$ and $\mathbb{P}(B_i)$ are independent of $i$. 

\subsubsection{Estimation of $\mathbb{P}(A_i)$}

For the estimation of $\mathbb{P}(A_i)$, consider two events $A_i$ and $A_j$ with $\dist(i,j)\gg \ell'$ such that they are independent. Suppose both $A_i$ and $A_j$ occur for a random disorder realization $\omega$. Recall that, as shown in Sec.~\ref{subsubsec: event A}, this implies that the state $|\Omega_\omega\rangle$ has an order parameter between $I_i$ and $I_j$. More explicitly, there exist odd operators $O_i[\omega]$ and $O_j[\omega]$ supported on $I_i$ and $I_j$ respectively, that depend on the disorder realization $\omega$, such that
\begin{equation}
    \left|\langle O_i[\omega] O_j[\omega]\rangle_{\Omega_\omega}\right| \ge 2^{-4\ell-3}\varepsilon_\ell^2.
\end{equation}
For small $\eta$, the probability that such an order parameter exists is correspondingly small. More concretely, by the definition of order parameter, since the ensemble does not have an $(\eta,\ell)$ order parameter, for any fixed choice of $O_i$ and $O_j$, we have
\begin{equation}
    \mathbb{E}_\omega \left[\left|\langle O_i O_j\rangle_{\Omega_\omega}\right|\right] < \eta.
\end{equation}
Recall Markov's inequality: if $X$ is a nonnegative random variable and $a>0$, then 
\begin{equation}
\mathbb{P}(X\ge a)\le \mathbb{E}[X]/a.
\end{equation} 
This can be verified by noting that 
\begin{equation}
    \begin{split}
\mathbb{E}[X]
= &\mathbb{E} \big[X\mathbf{1}_{\{X<a\}}\big] + \mathbb{E} \big[X \mathbf{1}_{\{X\ge a\}}\big]
\ge a \mathbb{E} \big[\mathbf{1}_{\{X\ge a\}}\big]\\
=& a \mathbb{P}(X\ge a).
    \end{split}
\end{equation}
For any operators $O_i$ and $O_j$, applying Markov's inequality to the random variable $X=|\langle O_i O_j\rangle_{\Omega_\omega}|$ with threshold 
\begin{equation}
a=\varepsilon_0\coloneqq2^{-4\ell-3}\varepsilon_\ell^2/3, 
\end{equation}
we have
\begin{equation}
\mathbb{P}\left(\left|\langle O_i O_j\rangle_{\Omega_\omega}\right| \ge \varepsilon_0\right) < \frac{\eta}{\varepsilon_0}=\frac{2^{4\ell+3}\cdot 3\eta}{\varepsilon_\ell^2}.
\end{equation}

Thus, for any fixed odd operators $O_i$ and $O_j$, the probability that they form an order parameter is upper bounded. On the other hand, if $A_i$ and $A_j$ occur, there must exist odd operators $O_i[\omega]$ and $O_j[\omega]$ that form an order parameter, though these operators may depend on the disorder realization $\omega$. Therefore, to estimate $\mathbb{P}(A_i)$, we need to take a union bound over all possible choices of $O_i$ and $O_j$. 

To this end, let $\mathcal{N}_i$ be an $\varepsilon_0$-net~\cite{papaspiliopoulos2020high} for local operators on $I_i$: a finite subset of the operator unit ball such that every local operator $O$ on $I_i$ with $\opnorm{O}\leq 1$ lies within operator-norm distance $\varepsilon_0$ of some $O'\in\mathcal{N}_i$ (see Appendix~\ref{app:epsilon_net}). The key observation is that whenever $A_i\cap A_j$ occurs, the $\omega$-dependent witnesses $O_i[\omega]$, $O_j[\omega]$ with $|\langle O_i[\omega] O_j[\omega]\rangle_{\Omega_\omega}|\ge 3\varepsilon_0$ are each approximated by some \emph{fixed} net element $O_i'\in\mathcal{N}_i$, $O_j'\in\mathcal{N}_j$ satisfying $\opnorm{O_i'-O_i[\omega]},\opnorm{O_j'-O_j[\omega]}\le\varepsilon_0$. By the triangle inequality, this fixed pair still witnesses the order parameter:
\begin{equation}
    |\langle O_i' O_j'\rangle_{\Omega_\omega}| \ge |\langle O_i[\omega] O_j[\omega]\rangle_{\Omega_\omega}| - 2\varepsilon_0 \ge \varepsilon_0.
\end{equation}
Thus $A_i\cap A_j$ implies that some pair in $\mathcal{N}_i\times\mathcal{N}_j$ achieves $|\langle O_i' O_j'\rangle_{\Omega_\omega}|\ge\varepsilon_0$. Since $|I_i|=\ell$, as shown in  Appendix~\ref{app:epsilon_net}, each net has size 
\begin{equation}
N_\ell(\varepsilon_0)\leq\bigl(1+\frac{2^{\ell/2+1}}{\varepsilon_0}\bigr)^{2^{2\ell+1}}
\end{equation}
independently of the site index, a union bound over all $N_\ell(\varepsilon_0)^2$ pairs gives
\begin{equation}
    \begin{split}
        \mathbb{P}(A_i\cap A_j) &\le \mathbb{P}\!\left(\exists\, O_i'\in\mathcal{N}_i,\, O_j'\in\mathcal{N}_j:\, \bigl|\langle O_i' O_j'\rangle_{\Omega_\omega}\bigr|\ge\varepsilon_0\right) \\
        &\le \frac{2^{4\ell+3}\cdot 3\eta}{\varepsilon_\ell^2}\,N_\ell(\varepsilon_0)^2.
    \end{split}
\end{equation}

Since $A_i$ and $A_j$ are independent, we have $\mathbb{P}(A_i\cap A_j)=\mathbb{P}(A_i)^2$, so
\begin{equation}
    \mathbb{P}(A_i)\;<\;2^{2\ell+1} \sqrt{6\eta}\;\frac{N_\ell\left(2^{-4\ell-3}\varepsilon_\ell^{ 2}/3\right)}{\varepsilon_\ell} .
\end{equation}

\subsubsection{Estimation of $\mathbb{P}(B_i)$}
We now estimate $\mathbb{P}(B_i)$. As in the $A_i$ case, we consider three independent events $B_i$, $B_j$, $B_k$ with $\dist(i,j),\dist(j,k),\dist(i,k)\gg\ell$. If all three occur, there exist $\omega$-dependent disorder-parameter witnesses in at least two of the regions $I_i^-$, $I_j^-$, $I_k^-$; taking these to be $I_i^-$ and $I_j^-$ without loss of generality, there exist operators $O_i^-[\omega]$ and $O_j^-[\omega]$ supported on $I_i^-$ and $I_j^-$ respectively such that
\begin{equation}
    \left|\langle O_i^-[\omega] O_j^-[\omega] S_{J\cup I_j^-}\rangle_{\Omega_\omega}\right| \ge \frac{1}{128},
\end{equation}
where $J$ is the interval between $I_i^-$ and $I_j^-$ in the clockwise direction. Since the ensemble has no $(\eta',\ell')$ disorder parameter, for any fixed $O_i^-$, $O_j^-$,
\begin{equation}
    \mathbb{E}_\omega\!\left[\left|\langle O_i^- O_j^- S_{J\cup I_j^-}\rangle_{\Omega_\omega}\right|\right] < \eta',
\end{equation}
and Markov's inequality with threshold $\varepsilon'_0\coloneqq 1/384$ gives
\begin{equation}
    \mathbb{P}\!\left(\left|\langle O_i^- O_j^- S_{J\cup I_j^-}\rangle_{\Omega_\omega}\right| \ge \varepsilon'_0\right) < 384\eta'.
\end{equation}
As in the $A_i$ case, the witnesses $O_i^-[\omega]$, $O_j^-[\omega]$ depend on $\omega$, so we introduce $\varepsilon'_0$-nets $\mathcal{N}_i^-$, $\mathcal{N}_j^-$ for local operators on $I_i^-$, $I_j^-$, each of size 
\begin{equation}
N_{\ell'}(\varepsilon'_0)\leq\bigl(1+\frac{2^{\ell'/2+1}}{\varepsilon'_0}\bigr)^{2^{2\ell'+1}}.
\end{equation}
Whenever $B_i\cap B_j\cap B_k$ occurs, the witnesses are approximated by fixed net elements $O_i^{-\prime}\in\mathcal{N}_i^-$, $O_j^{-\prime}\in\mathcal{N}_j^-$ with $\opnorm{O_i^{-\prime}-O_i^-[\omega]},\opnorm{O_j^{-\prime}-O_j^-[\omega]}\le\varepsilon'_0$, and the triangle inequality gives
\begin{equation}
    \left|\langle O_i^{-\prime} O_j^{-\prime} S_{J\cup I_j^-}\rangle_{\Omega_\omega}\right| \ge \frac{1}{128} - 2\varepsilon'_0 = \varepsilon'_0.
\end{equation}
Thus $B_i\cap B_j\cap B_k$ implies that some pair in $\mathcal{N}_i^-\times\mathcal{N}_j^-$ achieves threshold $\varepsilon'_0$. A union bound over all $N_{\ell'}(\varepsilon'_0)^2$ pairs, combined with independence $\mathbb{P}(B_i\cap B_j\cap B_k)=\mathbb{P}(B_i)^3$, gives
\begin{equation}
    \mathbb{P}(B_i) < \left(384\eta'\, N_{\ell'}\!\left(\tfrac{1}{384}\right)^2\right)^{1/3}.
\end{equation}

\subsubsection{Main trade-off result}\label{subsubsec:tradeoff-gapped}
We are now ready to prove our main trade-off result.
\begin{proposition}\label{prop:tradeoff}
    Let $f$, $g$, $\varepsilon_\ell$, $N_\ell$, and $\lambda_{1/4}$ be as defined below. A gapped ensemble $\{|\Omega_\omega\rangle\}$ has either an $(\eta,\ell)$ order parameter or an $(\eta',\ell')$ disorder parameter, provided $\ell$ and $\ell'$ satisfy
    \begin{enumerate}
        \item $f(\ell) \le \frac{1}{128}$,
        \item $\ell'\geq \lambda_{1/4}$,
        \item $f(\ell') \le \frac{(2^{-2\ell-1}\varepsilon_\ell)^2}{2}$,
        \item $(6+\sqrt{2})\sqrt{g(\ell')} + g(\ell') \le 2^{-2\ell-1}\varepsilon_\ell$.
    \end{enumerate}
    Moreover, $\eta$ and $\eta'$ satisfy the trade-off relation
    \begin{equation}
    2^{2\ell+1} \sqrt{6\eta}\;\frac{N_\ell\!\left(2^{-4\ell-3}\varepsilon_\ell^{2}/3\right)}{\varepsilon_\ell} + \left(384\eta' N_{\ell'}\!\left(\tfrac{1}{384}\right)^{\!2}\right)^{1/3} \geq 1.
    \end{equation}
    Here the auxiliary quantities are defined as follows. The functions $f$ and $g$ decay exponentially,
    \begin{align}
        f(\ell) &= \mathrm{poly}(\ell,\epsilon^{-1})\,e^{-c\epsilon\ell}, \\
        g(\ell) &= \mathrm{poly}(\ell,\epsilon^{-1})\,e^{-c_1\epsilon\ell},
    \end{align}
    where $c, c_1>0$ are numerical constants and $\epsilon$ is the uniform gap lower bound for the ensemble. The quantity $\varepsilon_\ell\in(0,1)$ solves
    \begin{equation}
        \ell+2H_2(\varepsilon_\ell/2)=\frac{1}{32},
    \end{equation}
    $N_\ell(\varepsilon)$ satisfies
    \begin{equation}
        N_\ell(\varepsilon)\leq \left(1+\frac{2^{\ell/2+1}}{\varepsilon}\right)^{2^{2\ell+1}},
    \end{equation}
    and $\lambda_{1/4}=\tilde{O}(1/\varepsilon^2)$ depends only on the gap.
\end{proposition}

\begin{proof}
    We prove by contradiction. Assume that the system has neither an $(\eta,\ell)$ order parameter nor an $(\eta',\ell')$ disorder parameter, then the probability of the events $A_i$ and $B_i$ are upper bounded:
    \begin{align}
        \mathbb{P}(A_i) &< 2^{2\ell+1} \sqrt{6\eta}\;\frac{N_\ell\left(2^{-4\ell-3}\varepsilon_\ell^{ 2}/3\right)}{\varepsilon_\ell}, \\
        \mathbb{P}(B_i) &< \left(384\eta' N_{\ell'}(\frac{1}{384})^2\right)^{1/3}.
    \end{align}

    However, since for any $i$ at least one of $A_i$ and $B_i$ must occur, we have $\mathbb{P}\left(A_i\right)+\mathbb{P}\left(B_i\right) \ge 1$, thus

    \begin{equation}
    2^{2\ell+1} \sqrt{6\eta}\;\frac{N_\ell\left(2^{-4\ell-3}\varepsilon_\ell^{ 2}/3\right)}{\varepsilon_\ell} + \left(384\eta' N_{\ell'}(\frac{1}{384})^2\right)^{1/3}>1,
    \end{equation}
    a contradiction. Therefore, the system must have either an $(\eta,\ell)$ order parameter or an $(\eta',\ell')$ disorder parameter, which completes the proof.
\end{proof}
It is worth noting that the conditions on $\ell$ and $\ell'$ can be satisfied simultaneously when $\ell$ and $\ell'$ are sufficiently large $O(1)$ constants, since $g(\ell)$ and $f(\ell)$ decay exponentially. In this case, all factors in the trade-off relation are $O(1)$ constants, so at least one of $\eta$ and $\eta'$ must be $O(1)$, which means that the system must have either an $O(1)$ order parameter or an $O(1)$ disorder parameter.
\subsection{Mutually exclusivity of order and disorder}\label{subsec:exclusion}

We have already shown that a gapped ensemble must have either an order parameter or a disorder parameter. We now show that these two possibilities are mutually exclusive: the existence of an order parameter rules out the existence of a disorder parameter. Therefore, a gapped ensemble has exactly one of the two.

To this end, we recall that such constraints have been shown for the clean case without disorder average in Ref.~\cite{Levin_2020}, as demonstrated by the following Lemma:
\begin{lemma}[Theorem 3, Ref.~\cite{Levin_2020}]\label{lem:Levin_thm3}
    If an Ising-symmetric ground state $|\Omega\rangle$ with gap $\epsilon$ in its symmetry sector has a $(\delta, \ell)$ order parameter defined at two points $i_1, i_3$ and a $(\delta, \ell)$ disorder parameter defined at two points $i_2, i_4$ with $i_1<i_2<i_3<i_4$, then
    \begin{equation}
        \min _{k, l}\left[\operatorname{dist}\left(i_k, i_l\right)\right] \leq 2 \ell+\tilde{\mathcal{O}}\left(\frac{\log \delta^{-1}}{\epsilon}\right)
    \end{equation}
\end{lemma}
For a translationally invariant state $|\Omega\rangle$, this suffices to show that if $|\Omega\rangle$ has an $O(1)$ order parameter, then it cannot have an $O(1)$ disorder parameter, and vice versa. Assume the contrary, we can take $i_{1,2,3,4}$ to be evenly spaced in the ring geometry, such that
\begin{equation}
    \min _{k, l}\left[\operatorname{dist}\left(i_k, i_l\right)\right] = \frac{L}{4},
\end{equation}
where $L$ is the system size. For sufficiently large $L$, this contradicts the bound in Lemma~\ref{lem:Levin_thm3}, thus at least one of the order parameter and disorder parameter must be absent.

Note that without additional structure for the Hamiltonians $\{H_\omega\}$ in the disorder ensemble, this mutual exclusivity cannot be generalized to the case with disorder average. To see this, consider an ensemble $\{H_1, H_2\}$ where $H_1$ is in the Ising ordered phase and $H_2$ is in the disordered phase, then the ensemble has both an order parameter and a disorder parameter, even though each individual Hamiltonian has only one of the two. 

However, we now argue that this does not happen in our setting. Recall that we assumed $H_\omega$ to take the form
\begin{equation}
    H_\omega=H_0+\sum_i (v_i^\omega \mathcal{O}_i+\hc),
\end{equation}
and that $H_\omega$ remains gapped in every symmetry sector with a uniform gap $\epsilon$ for all $\omega$. We further assume that $|\Omega_\omega\rangle$ lies in the same phase for every $\omega$. This is automatically guaranteed when $v_i^\omega$ varies continuously, since one can then interpolate between any two ground states via a quasi-adiabatic evolution. This same-phase condition rules out the counterexample above.

We now prove the mutual exclusivity for a gapped ensemble. Assume that the ensemble has a $(\delta,\ell)$ order parameter and an $(\delta,\ell)$ disorder parameter, with $\delta$ and $\ell$ being $O(1)$ constants. By definition, there exists odd operators $O_{i_1}$ and $O_{i_3}$ supported near sites $i_1$ and $i_3$ respectively, such that
\begin{equation}
    \mathbb{E}_\omega\left(\left|\langle O_{i_1} O_{i_3}\rangle_{\Omega_\omega}\right|\right)\ge \delta.
\end{equation}
Similarly, there exists an operator $O_{i_2}$ and $O_{i_4}^-$ supported near sites $i_2$ and $i_4$ respectively, such that
\begin{equation}
    \mathbb{E}_\omega\left(\left|\langle O_{i_2} O_{i_4} \prod_{k=i_2+1}^{i_4} S_k \rangle_{\Omega_\omega}\right|\right) \ge \delta.
\end{equation}

 Consider the event $A$ that 
 \begin{equation}
    \left|\langle O_{i_1} O_{i_3}\rangle_{\Omega_\omega}\right|\ge \frac{\delta}{2},
 \end{equation}
 \ie the state $|\Omega_\omega\rangle$ has a $(\delta/2,\ell)$ order parameter defines on two points $i_1$ and $i_3$, and the event $B$ that
 \begin{equation}
    \left|\langle O_{i_2} O_{i_4} \prod_{k=i_2+1}^{i_4} S_k \rangle_{\Omega_\omega}\right| \ge \frac{\delta}{2},
 \end{equation}
 \ie the state $|\Omega_\omega\rangle$ has a $(\delta/2,\ell)$ disorder parameter defines on two points $i_2$ and $i_4$. 
 
 Note that while $A$ only depends on the reduced density matrix near site $i_1$ and $i_3$, the event $B$ involves a non-local operator supported on the region between $i_2$ and $i_4$. Therefore, unlike in previous sections, the events $A$ and $B$ are not necessarily independent, but we can still show that both events have an $O(1)$ probability to occur.

 Concretely, since the expectation value of the order parameter is at most 1, we have
 \begin{equation}
    \delta \le \mathbb{E}_\omega\left(\left|\langle O_{i_1} O_{i_3}\rangle_{\Omega_\omega}\right|\right) \le \mathbb{P}(A)\cdot 1+ \mathbb{P}(\overline{A})\cdot \frac{\delta}{2},
 \end{equation}
thus $\mathbb{P}(A)\ge \delta/(2-\delta)$. Similarly, $\mathbb{P}(B)\ge \delta/(2-\delta)$. Therefore, there exists two states $|\Omega_1\rangle$ and $|\Omega_2\rangle$ such that $|\Omega_1\rangle$ has a $(\delta/2,\ell)$ order parameter defined on two points $i_1$ and $i_3$, and $|\Omega_2\rangle$ has a $(\delta/2,\ell)$ disorder parameter defined at two points $i_2$ and $i_4$. Since they are in the same phase, there exists a quasi-adiabatic evolution $U$ that connects $|\Omega_1\rangle$ and $|\Omega_2\rangle$. Since $U$ is symmetric and can be approximated by a finite-depth circuit~\cite{Hastings_2005,Bravyi_2010,Bachmann_2011,Haah_2021,yi2025universaldecayconditionalmutual}, it will only dress the order parameter $O_{i_1}$ and $O_{i_3}$ quasi-locally, and $|\Omega_2\rangle$ will also have a $(O(\delta),O(\ell+\xi))$ order parameter defined on two points $i_1$ and $i_3$, where $\xi$ is the correlation length, in addition to the $(\delta/2,\ell)$ disorder parameter defined at two points $i_2$ and $i_4$.

This contradicts Lemma~\ref{lem:Levin_thm3}, thus at least one of the order parameter and disorder parameter must be absent, which completes the proof.
\subsection{Parity of disorder parameter}\label{subsec:disorderparity}
We now finish with the final ingredient for the proof of the main theorem, which is that if the gapped ensemble does have a disorder parameter, then it must have even parity. This can be shown by a similar argument as in the proof of the mutual exclusivity. 

To this end, we recall that such constraints have been shown for the clean case without disorder average in Ref.~\cite{Levin_2020}, as demonstrated by the following Lemma:
\begin{lemma}[Theorem 4, Ref.~\cite{Levin_2020}]\label{lem:Levin_thm4}
    If an Ising-symmetric ground state $|\Omega\rangle$ with gap $\epsilon$ in its symmetry sector has a $(\delta, \ell)$  disorder parameter with odd parity under $S$, defined at $i_1, i_2, i_3, i_4$, then
    \begin{equation}
        \min _{k, l}\left[\operatorname{dist}\left(i_k, i_l\right)\right] \leq 2 \ell+\tilde{\mathcal{O}}\left(\frac{\log \delta^{-1}}{\epsilon}\right)
    \end{equation}
\end{lemma}
Similar to the previous section, for a translational invariant state $|\Omega\rangle$, this suffices to show that if $|\Omega\rangle$ has an $O(1)$ disorder parameter, it cannot be odd under $S$, since otherwise we can take $i_{1,2,3,4}$ to be evenly spaced in the ring geometry, and their minimum distance should not have an $O(1)$ upper bound.

We now generalize this result to the case with disorder average. Assume that the ensemble has a $(\delta, \ell)$ disorder parameter. Choose $2n$ points $i_1, i_2, \cdots, i_{2n}$ evenly spaced in the ring geometry, with $n$ being an $O(1)$ constant to be chosen later, such that $\dist(i_k, i_l)\gg \ell$ for all $k\neq l$. By definition, there exist operators $O_{i_k}$ supported near site $i_k$, such that
\begin{equation}
    \mathbb{E}_\omega\left(\left|\langle O_{i_k} O_{i_k+n} \prod_{m=i_k+1}^{i_{k+n}} S_m \rangle_{\Omega_\omega}\right|\right)\ge \delta.
\end{equation}
We now define a series of events $B_k'$ for $k=1,2,\cdots,n$ such that
\begin{equation}
    \left|\langle O_{i_k} O_{i_k+n} \prod_{m=i_k+1}^{i_{k+n}} S_m \rangle_{\Omega_\omega}\right|\ge \frac{\delta}{2}.
\end{equation}
\ie the state $|\Omega_\omega\rangle$ has a $(\delta/2,\ell)$ disorder parameter defined on two points $i_k$ and $i_{k+n}$. Note that these events depend non-locally on $|\Omega_\omega\rangle$, so they are not necessarily independent, but similar to the previous section, we can still show that their probability is an $O(1)$ constant. Concretely, since the expectation value of the disorder parameter is at most 1, we have
\begin{equation}
    \begin{split}
        \delta &\le \mathbb{E}_\omega\left(\left|\langle O_{i_k} O_{i_k+n} \prod_{m=i_k+1}^{i_{k+n}} S_m \rangle_{\Omega_\omega}\right|\right)\\
         \le &\mathbb{P}(B_k')\cdot 1+ \mathbb{P}(\overline{B_k'})\cdot \frac{\delta}{2},   
    \end{split}
\end{equation}
thus $\mathbb{P}(B_k')\ge \delta/(2-\delta)$. 

Now we proceed to show that there exists two events $B_k'$ and $B_l'$ such that their intersection has positive probability. We can consider the indicator variables $X_k=\mathbf{1}_{B_k'}$ for $k=1,2,\cdots,n$. Since $\mathbb{P}(B_k')\ge \delta/(2-\delta)$, we have $\mathbb{E}_\omega(X_k)=\mathbb{P}(B_k')\ge \delta/(2-\delta)$. Defining $S=\sum_{k=1}^n X_k$, we have
\begin{equation}
    \begin{split}
        \mathbb{E}_\omega(S^2)&=\sum_i \mathbb{E}_\omega(X_i)+\sum_{i\neq j} \mathbb{E}_\omega(X_iX_j)\\
        &\leq \mathbb{E}_\omega(S)+ \frac{n(n-1)}{2}\max_{kl}\mathbb{P}(B_k'\cap B_l')\;.
    \end{split}
\end{equation}
On the other hand, we have
\begin{equation}
    \mathbb{E}_\omega(S^2)=\mathbb{E}_\omega(S)^2+\mathrm{Var}(S)\geq \mathbb{E}_\omega(S)^2,
\end{equation}
where $\mathrm{Var}(S)$ is the variance of $S$. Thus we have
\begin{equation}
    \max_{k,l}\mathbb{P}(B_k'\cap B_l')\geq \frac{2}{n(n-1)}\left( \mathbb{E}_\omega(S)^2-\mathbb{E}_\omega(S)\right).
\end{equation}
Since 
\begin{equation}
\mathbb{E}_\omega(S)=\sum_k \mathbb{E}_\omega(X_k)\ge n\delta/(2-\delta),
\end{equation}
if we take $n>\frac{2-\delta}{\delta}$, then $\mathbb{P}(B_k'\cap B_l')>0$ is an $O(1)$ constant for some $k\neq l$. Therefore, there exists a state $|\Omega_\omega\rangle$ that has a $(\delta/2,\ell)$ disorder parameter with odd parity defined at two pairs of points $i_k$, $i_{k+n}$, and $i_l$, $i_{l+n}$ . By Lemma~\ref{lem:Levin_thm4}, this is not possible for sufficiently large $L$, a contradiction. Therefore, the disorder parameter must be even under $S$, which completes the proof.

\section{Rare region effects}\label{sec:Rareregion}
In this section, we discuss the rare region effect and demonstrate that it does not invalidate our main theorem. Specifically, we show that even for a nearly gapped ensemble with the presence of rare regions, the system still has one and only one of the two: an $O(1)$ order parameter or an $O(1)$ disorder parameter with even parity.

The key idea of the proof is to also denote the absence of rare region effect as an event and control the probability correspondingly. More concretely, we can introduce a new series of events $C_i$ such that $C_i$ occurs if there is no rare region effect in the region near site $i$. Concretely, $C_i$ occurs if for the state $|\Omega_\omega\rangle$, the rare region effect can be represented by a unitary supported outside of the region $I_i^-I_iI_i^+$, \ie
\begin{equation}
    |\Omega_\omega\rangle = U_{I_i'} |\Omega\rangle,
\end{equation}
where $|\Omega\rangle$ is a state with a uniform gap $\epsilon$ in the symmetry sector, and $U_{I_i'}$ is a symmetric unitary operator supported in the region $I_i'$, which is the complement of $I_i^-I_iI_i^+$. By the LPPL principle, the event $C_i$ should only depend on the disorder realization near the site $i$, so it is independent of the events $A_j$, $B_j$ and $C_j$ for $j$ with $\dist(i,j)\gg \ell'$. Moreover, we expect {for all $O(1)$ constants $\ell$ and $\ell'$,
\begin{equation}
    \mathbb{P}(C_i)>p_0,
\end{equation}
for some $O(1)$ constant $p_0$, depending on $\ell, \ell'$ that is very close to 1.
}

{The fact that $\mathbb{P}(C_i)$ is an $O(1)$ constant, independent of the system size $L$, is an important feature of nearly gapped ensemble, with Griffiths-like rare regions that have exponentially small probability. In contrast, for the infinite-randomness fixed points\cite{Fisher1994,PhysRevB.51.6411}, $\mathbb{P}(C_i)$ vanishes as $L\to\infty$, and our results below will not apply.}

Since the proof closely follows the gapped case without rare regions, we will only sketch it, focusing on verifying that the key lemmas remain valid even in the presence of rare regions.

\subsection{Trade-off for nearly gapped ensembles}\label{subsec: tradeoff-rare}
We start by showing that for any nearly gapped ensemble $\{|\Omega_\omega\rangle\}$, it must have either an $O(1)$ order parameter or an $O(1)$ disorder parameter. We define the events $A_i$ and $B_i$ as in Sec.~\ref{subsec: events}. The argument follows the same Markov's inequality and $\varepsilon$-net strategy as in Sec.~\ref{subsec:tradeoff_proof}, bounding $\mathbb{P}(A_i\cap C_i)$ and $\mathbb{P}(B_i\cap C_i)$ separately. The main work is to verify that Lemma~\ref{lem:factorization}, Lemma~\ref{lem:Levin3}, and Lemma~\ref{lem:Levin2} each generalize to the rare-region setting; once this is done, the rest of the proof is unchanged.

\subsubsection{Estimation of $\mathbb{P}(A_i\cap C_i)$}
We now estimate $\mathbb{P}(A_i\cap C_i)$. Suppose both $A_i$ and $C_i$ occur for some disorder realization $\omega$, we first show that $|\Omega_\omega\rangle$ is strongly-ordered on region $I_i$ and $I_i'$. The proof follows the same chain as in Sec.~\ref{subsec: events}: Lemma~\ref{lem:Fannes_Cond_entropy} converts the mutual information lower bound $I_\Omega(I_i:I_i')>1/32$ into a trace-norm lower bound $\|\Omega_{I_i}\otimes\Omega_{I_i'}-\Omega_{I_iI_i'}\|_1>\varepsilon_\ell$, and Lemma~\ref{lem:brandao14} converts this into the existence of operators on $I_i$ and $I_i'$ with a large connected correlator. Neither step uses the gap, so both are unaffected by rare regions. It therefore suffices to generalize Lemma~\ref{lem:factorization}, which bounds the connected correlator of even operators and relies on the spectral gap:
{
\renewcommand{\thelemma}{\getrefnumber{lem:factorization}$'$}
\renewcommand{\theHlemma}{lem.factorization.prime}
\addtocounter{lemma}{-1}
\begin{lemma}\label{lem:factorization-prime}
Let $O_1$ and $O_2$ be operators that are even under $S$, have norm at most $1$, and are supported on intervals $I_1$ and $I_2$ separated by distance at least $\ell'$. Then for any state $|\Omega'\rangle=U_{I_2}|\Omega\rangle$ such that $U_{I_2}$ is a (not necessarily local) unitary supported on the region $I_2$, and $|\Omega\rangle$ has a uniform gap $\epsilon$ in the symmetry sector, we have
\begin{equation}
    \left|
        \langle O_1O_2\rangle_{\Omega'}
        - \langle O_1\rangle_{\Omega'}\langle O_2\rangle_{\Omega'}
    \right|
    \le
    (6+\sqrt{2})\sqrt{g(\ell')} + g(\ell'),
\end{equation}
where
\begin{equation}
    g(\ell')=\mathrm{poly}(\ell',\epsilon^{-1})e^{-c_1\epsilon \ell'}
\end{equation}
for some numerical constant $c_1>0$, and $\epsilon$ is the energy gap within the even subspace.
\end{lemma}
}
\begin{proof}
Note that
\begin{equation}
    \begin{split}
            &\left|
        \langle O_1O_2\rangle_{\Omega'}
        - \langle O_1\rangle_{\Omega'}\langle O_2\rangle_{\Omega'}
    \right|\\
    =&\left|
        \langle O_1 U_{I_2}^\dagger O_2 U_{I_2} \rangle_{\Omega}
        - \langle O_1\rangle_{\Omega}\langle U_{I_2}^\dagger O_2 U_{I_2}\rangle_{\Omega}
    \right|
    \end{split}
\end{equation}
We can apply Lemma.~\ref{lem:factorization} to the state $|\Omega\rangle$ and the operators $O_1$ and $U_{I_2}^\dagger O_2 U_{I_2}$. This gives the desired bound, which completes the proof of Lemma.~\ref{lem:factorization-prime}.
\end{proof}

Therefore, similar to the proof in Sec.~\ref{subsec: events}, by Lemma.~\ref{lem:Fannes_Cond_entropy},~\ref{lem:brandao14}, and~\ref{lem:factorization-prime}, we have if the events $A_i$ and $C_i$ occur, then there exist an $O(1)$ order parameter on $I_i$ and $I_i'$. 

We now proceed to promote this non-local order parameter to a local order parameter defined on two points. Suppose that four events $A_i$, $A_j$, $C_i$ and $C_j$ occur, with $\dist(i,j)\gg \ell'$. To convert the order parameters on $I_i$ and $I_j$ to an order parameter for the sites $i$ and $j$, we need to check Lemma.~\ref{lem:Levin3}, which can be generalized as follows with the presence of rare region effect:
{
\renewcommand{\thelemma}{\getrefnumber{lem:Levin3}$'$}
\renewcommand{\theHlemma}{lem.Levin3.prime}
\addtocounter{lemma}{-1}
\begin{lemma}\label{lem:Levin3-prime}
Let $O_1,O_1',O_2,O_2'$ be operators supported on intervals $I_1,I_1',I_2,I_2'$ respectively, all either odd or even under $S$, with norm at most $1$. Denote $I_1^-$ and $I_1^+$ as the intervals between $I_1$ and $I_1'$ in the counter-clockwise and clockwise direction respectively, and define $I_2^-$ and $I_2^+$ similarly.
Define
\begin{equation}
    \ell_{\min}=\min\bigl(\dist(I_1,I_1'),\dist(I_1',I_2'),\dist(I_1,I_2)\bigr).
\end{equation}
Then for any state $|\Omega'\rangle=U|\Omega\rangle$ such that $U$ is a (not necessarily local) symmetric unitary operator supported outside of $I_1^-I_1I_1^+$ and $I_2^-I_2I_2^+$, and $|\Omega\rangle$ has a uniform gap $\epsilon$ in the symmetry sector, we have
\begin{equation}
    \left|\langle O_1 O_2\rangle_{\Omega'}\right| \ge \left|\langle O_1 O_1'\rangle_{\Omega'}\right| \left|\langle O_2 O_2'\rangle_{\Omega'}\right| - f(\ell_{\min}),
\end{equation}
and
\begin{equation}
    \begin{split}
    &\left|\left\langle O_1 O_2 S_{J\cup I_2}\right\rangle_{\Omega'}\right|\\
     \ge &\left|\left\langle O_1 O_1' S_{I_1^+}\right\rangle_{\Omega'}\right| \left|\left\langle O_2 O_2' S_{I_2^+}\right\rangle_{\Omega'}\right| - f(\ell_{\min}),
    \end{split}
\end{equation}
where $f(\ell)=\mathrm{poly}(\ell,\epsilon^{-1})e^{-c\epsilon \ell}$, $J$ is the interval between $I_1$ and $I_2$ in the clockwise direction.
\end{lemma}
}
\begin{proof}
Note that
\begin{equation}
    \begin{split}
    &\left|\langle O_1 O_2\rangle_{\Omega'}\right|
    =\left|\langle O_1 O_2 \rangle_{\Omega}\right|\\
    \ge& \left|\langle O_1 U^\dagger O_1' U\rangle_{\Omega}\right| \left|\langle O_2 U^\dagger O_2' U\rangle_{\Omega}\right| - f(\ell_{\min})\\
    = &\left|\langle O_1 O_1'\rangle_{\Omega'}\right| \left|\langle O_2 O_2'\rangle_{\Omega'}\right| - f(\ell_{\min}),  
    \end{split}
\end{equation}
where we have applied Lemma.~\ref{lem:Levin3} to the state $|\Omega\rangle$ and the operators $O_1$, $U^\dagger O_1' U$, $O_2$ and $U^\dagger O_2' U$. We can similarly show the second inequality by noticing that $U$ commutes with all three symmetry restriction operators $S_{I_1^+}$, $S_{I_2^+}$ and $S_{J\cup I_2}$. We thus complete the proof of Lemma.~\ref{lem:Levin3-prime}.
\end{proof}

Following the same estimation as in Sec.~\ref{subsec:tradeoff_proof}, we have if the events $A_i$, $A_j$, $C_i$ and $C_j$ occur, then there exist an $O(1)$ order parameter defined on two points $i$ and $j$. Therefore, if we assume that the ensemble does not have an $(\eta,\ell)$ order parameter, by the Markov's inequality and the $\epsilon$-net argument, we have 
\begin{equation}
    \mathbb{P}(A_i\cap C_i\cap A_j\cap C_j)<\frac{2^{4\ell+3}\cdot 3 \eta}{\varepsilon_\ell^{ 2}} N_\ell(2^{-4\ell-3}\varepsilon_\ell^{ 2}/3)^2.
\end{equation}
Since $A_i\cap C_i$ and $A_j\cap C_j$ are independent, we have 
\begin{equation}
    \mathbb{P}(A_i\cap C_i)\;<\;2^{2\ell+1} \sqrt{6\eta}\;\frac{N_\ell\left(2^{-4\ell-3}\varepsilon_\ell^{ 2}/3\right)}{\varepsilon_\ell} .
\end{equation}

\subsubsection{Estimation of $\mathbb{P}(B_i\cap C_i)$}
We now move on to estimate $\mathbb{P}(B_i\cap C_i)$. Recall that if $B_i$ occurs, the state $|\Omega_\omega\rangle$ will be weakly-disordered on the region $I_i^-$ and $I_i^+$. Furthermore, for the case without rare region effect, we can leverage Lemma.~\ref{lem:Levin2} to show that $|\Omega_\omega\rangle$ is also strongly-disordered on the region $I_i^-$ and $I_i^+$. We can then use the Markov's inequality and the $\epsilon$-net argument to show that the probability of $B_i$ is small. It then suffices to generalize Lemma.~\ref{lem:Levin2}, which shows that weak disorder parameters can be upgraded to strong disorder parameters. For completeness, we also include the analogous statement for order parameters, although it is not used for our purpose. 
{\renewcommand{\thelemma}{\getrefnumber{lem:Levin2}$'$}
\renewcommand{\theHlemma}{lem.Levin2.prime}
\addtocounter{lemma}{-1}
\begin{lemma}\label{lem:Levin2-prime}
Consider a state $|\Omega'\rangle=U|\Omega\rangle$ such that $U$ is a (not necessarily local) unitary supported inside $I_1$ or $I_2$ or $(I_1\cup I_2)^c$, and $|\Omega\rangle$ has a uniform gap $\epsilon$ in the symmetry sector. Then for every $\delta>0$, there exists a length $\lambda_\delta=\tilde{O} \left(1/\varepsilon^{2}\right)$ such that if $|\Omega'\rangle$ is $\delta$ weakly-ordered on $I_1,I_2$, then it is
$\delta/2$ strongly-ordered on $I_1,I_2$ for all $I_1,I_2$ that are separated by
a distance of at least $\lambda_\delta$. Likewise, if $|\Omega\rangle$ is $\delta$
weakly-disordered on $I_1,I_2$ with even (odd) parity, then it is $\delta/2$
strongly-disordered on $I_1,I_2$ with even (odd) parity for all $I_1,I_2$
separated by at least $\lambda_\delta$.
\end{lemma}
}
\begin{proof}
Note that if $|\Omega'\rangle$ is $\delta$ weakly-ordered on $I_1,I_2$, then there exist an operator $A$ supported on $I_1\cup I_2$, such that
\begin{equation}
    \left|\langle A\rangle_{\Omega'}\right|\ge \delta.
\end{equation}
Suppose that $U$ is supported in $(I_1\cup I_2)^c$, and thus commutes with $A$, we have
\begin{equation}
    \left|\langle A\rangle_{\Omega'}\right|=\left|\langle A\rangle_{\Omega}\right|\ge \delta,
\end{equation}
so $|\Omega\rangle$ is also $\delta$ weakly-ordered on $I_1,I_2$. By Lemma.~\ref{lem:Levin2}, $|\Omega\rangle$ is $\delta/2$ strongly-ordered on $I_1,I_2$ for all $I_1,I_2$ that are separated by a distance of at least $\lambda_\delta$. Since $U$ commutes with $O_1$ and $O_2$, $|\Omega'\rangle$ is also $\delta/2$ strongly-ordered on $I_1,I_2$.

Similarly, if $|\Omega'\rangle$ is $\delta$ weakly-disordered on $I_1,I_2$ with even (odd) parity, then there exist an operator $A$ supported on $I_1\cup I_2$, such that
\begin{equation}
    \left|\langle A S_{J}\rangle_{\Omega'}\right|\ge \delta,
\end{equation}
where $J$ is the interval between $I_1$ and $I_2$. Note that by the LPPL principle, $U$ should factorize into unitaries on $J$ and $J'$ respectively, where $J'=(I_1\cup I_2\cup J)^c$, \ie
\begin{equation}
    U=U_J U_{J'}.
\end{equation}
Since $U_{J}$ and $U_{J'}$ should be symmetric, they should commute with $S_J$, so we have
\begin{equation}
    \left|\langle A S_{J}\rangle_{\Omega'}\right|=\left|\langle A S_{J}\rangle_{\Omega}\right|\ge \delta,
\end{equation}
so $|\Omega\rangle$ is also $\delta$ weakly-disordered on $I_1,I_2$ with even (odd) parity. By Lemma.~\ref{lem:Levin2}, $|\Omega\rangle$ is $\delta/2$ strongly-disordered on $I_1,I_2$ with even (odd) parity for all $I_1,I_2$ separated by at least $\lambda_\delta$. Since $U$ commutes with $A$ and $S_J$, $|\Omega'\rangle$ is also $\delta/2$ strongly-disordered on $I_1,I_2$ with even (odd) parity.

On the other hand, suppose that $U$ is supported in $I_1$. If $|\Omega'\rangle$ is $\delta$ weakly-ordered on $I_1,I_2$, then there exist an operator $A$ supported on $I_1\cup I_2$, such that
\begin{equation}
    \left|\langle A\rangle_{\Omega'}\right|\ge \delta.
\end{equation}
Since $U$ is supported on $I_1$, for the state $|\Omega\rangle$
\begin{equation}
    \left|\langle U^\dagger AU\rangle_{\Omega}\right|\ge \delta,
\end{equation} 
rendering $|\Omega\rangle$ $\delta$ weakly-ordered on $I_1,I_2$. By Lemma.~\ref{lem:Levin2}, 
$|\Omega\rangle$ is $\delta/2$ strongly-ordered on $I_1,I_2$, conjugating by $U$ on the order parameters, we have $|\Omega'\rangle$ is also $\delta/2$ strongly-ordered on $I_1,I_2$. For the case with disorder parameter, the proof is also similar. Therefore, we complete the proof of Lemma.~\ref{lem:Levin2-prime}.
\end{proof}

We are now ready for the analysis of the events $B_i$. We have already shown by Lemma.~\ref{lem:Levin2-prime} that if the events $B_i$ occur, then the state $|\Omega_\omega\rangle$ is weakly-disordered on the region $I_i^-$ and $I_i^+$. By Lemma.~\ref{lem:Levin3-prime}, similar to the estimation as in Sec.~\ref{subsec:tradeoff_proof}, we have that suppose there are six events $B_{i,j,k}$ and $C_{i,j,k}$ occur, with $\dist(i,j),\dist(i,k),\dist(j,k)\gg \ell'$, then there exist an $O(1)$ disorder parameter defined on two of the three sites $i,j,k$. Following the same estimation as in Sec.~\ref{subsec:tradeoff_proof}, if the ensemble does not have an $(\eta',\ell')$ disorder parameter, by the Markov's inequality and the $\epsilon$-net argument, we have
\begin{equation}
    \mathbb{P}(B_i\cap C_i\cap B_j\cap C_j\cap B_k\cap C_k)<384\eta' N_{\ell'}(\frac{1}{384})^2.
\end{equation}
Since $B_i\cap C_i$, $B_j\cap C_j$ and $B_k\cap C_k$ are independent, we have
\begin{equation}
    \mathbb{P}(B_i\cap C_i)< \left(384\eta' N_{\ell'}(\frac{1}{384})^2\right)^{1/3}.
\end{equation}

\subsubsection{Trade-off relation}
We begin by generalizing the incompatibility condition for the gapped case without the rare region effect. Recall that in that case, the incompatibility between the events $\overline{A_i}$ and $\overline{B_i}$ is shown by Lemma.~\ref{lem:Levin2}. For the case with rare region effect, we can similarly show that if the event $\overline{B_i}\cap C_i$ occurs, then
\begin{equation}
     I_\Omega(I_i:I_i')>\frac{1}{32}
\end{equation}
which implies that the event $\overline{A_i}$ cannot occur. Therefore, we have $C_i\subset A_i\cup B_i$, and thus
\begin{equation}
    \mathbb{P}(C_i)\le \mathbb{P}(A_i\cap C_i)+\mathbb{P}(B_i\cap C_i).
\end{equation}

We are now ready to prove the trade-off relation for nearly gapped ensembles. Assume that the ensemble does not have an $(\eta,\ell)$ order parameter and an $(\eta',\ell')$ disorder parameter, we have
\begin{align}
        \mathbb{P}(A_i\cap C_i) &< 2^{2\ell+1} \sqrt{6\eta}\;\frac{N_\ell\left(2^{-4\ell-3}\varepsilon_\ell^{ 2}/3\right)}{\varepsilon_\ell}, \\
        \mathbb{P}(B_i\cap C_i) &< \left(384\eta' N_{\ell'}(\frac{1}{384})^2\right)^{1/3}.
\end{align}
On the other hand, we have $\mathbb{P}(C_i)>p_0$ is some $O(1)$ constant. Therefore, if we take
\begin{equation}
    \begin{split}
        &2^{2\ell+1} \sqrt{6\eta}\;\frac{N_\ell\!\left(2^{-4\ell-3}\varepsilon_\ell^{2}/3\right)}{\varepsilon_\ell} \\
        +\,&\left(384\eta'\, N_{\ell'}\!\left(\tfrac{1}{384}\right)^2\right)^{1/3} = p_0,
    \end{split}
\end{equation}
then we have $\mathbb{P}(C_i)> \mathbb{P}(A_i\cap C_i)+\mathbb{P}(B_i\cap C_i)$, a contradiction. Therefore, a nearly gapped ensemble must have either an $(\eta,\ell)$ order parameter or an $(\eta',\ell')$ disorder parameter, which completes the proof of the trade-off relation for nearly gapped ensembles.

\subsection{Mutual exclusivity for nearly gapped ensembles}\label{subsec: exclusion-rare}
{We now show that a nearly gapped ensemble cannot have both an $O(1)$ order parameter and an $O(1)$ disorder parameter.

Suppose for contradiction that both a $(\delta,\ell)$ order parameter and a $(\delta,\ell)$ disorder parameter exist. Define events $A$ and $B$ as in Sec.~\ref{subsec:exclusion}: $A$ is the event that a disorder realization $\omega$ yields a state $|\Omega_\omega\rangle$ satisfying
\begin{equation}
    \left|\langle O_{i_1} O_{i_3}\rangle_{\Omega_\omega}\right|\ge \frac{\delta}{2},
\end{equation}
and $B$ is the event that
\begin{equation}
    \left|\langle O_{i_2} O_{i_4} \prod_{k=i_2+1}^{i_4} S_k \rangle_{\Omega_\omega}\right| \ge \frac{\delta}{2}.
\end{equation}
Since the probabilistic argument of Sec.~\ref{subsec:exclusion} does not depend on whether rare regions are present, it applies unchanged and gives $\mathbb{P}(A)\ge \delta/(2-\delta)$ and $\mathbb{P}(B)\ge \delta/(2-\delta)$. 

We can now pick $i_{1,2,3,4}$ to be evenly spaced in the whole geometry, and let $R_{k}$ denote the event that there exists a rare region covering the site $i_k$ of size at least $L/12$. By the nearly gapped assumption,
\begin{equation}
    \mathbb{P}(R_{k})\sim p^{L/12}\;.
\end{equation}
Consider the event $R$ that there is a rare region that spans one of the four sites, \ie $R\coloneq \bigcup_{k} R_{k}$. By the union bound,
\begin{equation}
    \mathbb{P}(R)\leq \sum_{k}\mathbb{P}(R_{k})\sim O(p^{L/12})\;.
\end{equation}
We now show that events $A$ and $B$ are not merely artifacts of large rare regions smearing out local operators near $i_{1,2,3,4}$. By the union bound,
\begin{equation}
    \begin{split}
    \mathbb{P}(A\cap\overline{R})
    &\ge \mathbb{P}(A)-\mathbb{P}(R)\ge \frac{\delta}{2-\delta}-O(p^{L/12}),\\
    \mathbb{P}(B\cap\overline{R})
    &\ge \mathbb{P}(B)-\mathbb{P}(R)\ge \frac{\delta}{2-\delta}-O(p^{L/12}).
    \end{split}
\end{equation}
Taking the system size to be large, both right-hand sides give $O(1)$ constant lower bounds.

However, for any realization $\omega$ for which $A\cap \overline{R}$ occurs, there exists a $(\delta/2,L/12)$ order parameter on sites $i_{1,3}$ for some gapped state. Similarly, for any realization $\omega$ for which $B\cap \overline{R}$ occurs, there exists a $(\delta/2,L/12)$ disorder parameter on sites $i_{2,4}$ for some gapped state. This violates the absence of phase separation condition for our ensemble, a contradiction.

This completes the proof that a nearly gapped ensemble cannot have both an $O(1)$ order parameter and an $O(1)$ disorder parameter in the presence of rare regions.}

\subsection{Parity of disorder parameter for nearly gapped ensembles}\label{subsec: disorderparity-rare}
{We now show that any disorder parameter of a nearly gapped ensemble must parity even. 

The key idea is similar to Sec.~\ref{subsec:disorderparity}, and we treat the rare region effects as in Sec.~\ref{subsec: exclusion-rare}.

Suppose for contradiction that the ensemble has a $(\delta,\ell)$ disorder parameter with odd parity. Place $n$ pairs of sites $(i_k,i_{k+n})$ for $k=1,\ldots,n$ evenly in the geometry, each pair separated by $L/2$. Define $B_k'$ as the event that the realization $\omega$ yields a state satisfying
\begin{equation}
    \left|\left\langle O_{i_k} O_{i_{k+n}} \prod_{m=i_k+1}^{i_{k+n}} S_m \right\rangle_{\Omega_\omega}\right| \ge \frac{\delta}{2}.
\end{equation}
Since the probabilistic argument of Sec.~\ref{subsec:exclusion} remains valid, we have $\mathbb{P}(B_k')\ge\delta/(2-\delta)$ for each $k$. Choosing $n$ large enough that $\sum_k\mathbb{P}(B_k')>1$, the second-moment argument of Sec.~\ref{subsec:disorderparity} then produces indices $k\neq l$ with all four sites $i_k,i_{k+n},i_l,i_{l+n}$ pairwise well-separated with a distance $L/2n$, such that $\mathbb{P}(B_k'\cap B_l')=O(1)$.

Since $\mathbb{P}(B_k'\cap B_l')=O(1)$, these events cannot merely come from rare regions near the four sites. For each $j\in\{i_k,i_{k+n},i_l,i_{l+n}\}$ let $R_j$ denote the event that a rare region of size at least $L/(6n)$ covers site $j$, and let $R=\bigcup_j R_j$. By the nearly gapped assumption,
\begin{equation}
    \mathbb{P}(R_j)\sim p^{L/(6n)},
\end{equation}
and by the union bound,
\begin{equation}
    \mathbb{P}(R)\leq\sum_j\mathbb{P}(R_j)\sim O(p^{L/(6n)}).
\end{equation}
Hence
\begin{equation}
    \mathbb{P}(B_k'\cap B_l'\cap\overline{R})\ge\mathbb{P}(B_k'\cap B_l')-\mathbb{P}(R)>0
\end{equation}
for sufficiently large $L$. However, for any realization $\omega$ for which $B_k'\cap B_l'\cap\overline{R}$ occurs, the state $|\Omega_\omega\rangle$ has two odd-parity disorder parameters at the well-separated pairs $(i_k,i_{k+n})$ and $(i_l,i_{l+n})$, contradicting Lemma~\ref{lem:Levin_thm4}. Therefore the disorder parameter must be parity-even, which completes the proof.
}

\section{Applications}\label{sec: Application}
In this section, we discuss some applications of our main theorem. We first extend our theorems to the case of symmetry-protected topological (SPT) phase, and demonstrate a well-defined string order parameter for SPT in dirty systems in Sec.~\ref{subsec:aspt}. We then use our main theorem to show that there is a Lieb-Schultz-Mattis-type constraint for average ensemble in Sec.~\ref{subsec:disorderLSM}. In Sec.~\ref{subsec:JordanWigner} we apply our results to fermion chains through the Jordan-Wigner transform and obtain trade-off relations between different disorder parameters.

\subsection{String order parameter for  SPT in dirty systems}\label{subsec:aspt}

While traditionally the concept of SPT phases is defined for clean systems, it can be naturally extended to dirty systems~\cite{PhysRevX.13.031016,Ma_2025}. For the clean systems, the string order parameter~\cite{PhysRevB.45.304,P_rez_Garc_a_2008,Pollmann_2012,Haegeman_2012} serves as a powerful tool to characterize the SPT phases, which is defined as the expectation value of two operators connected by a string of restricted symmetry operators. Note that this coincides with our definition of disorder parameter in the clean case, so it is natural to ask whether our definition in the disordered case can also be used to characterize SPT phases. In this section, we show that this is indeed the case, and we can use our main theorem to show that there is a well-defined string order parameter for some disordered SPT phases.

For simplicity, we focus on the case of $1d$ SPT phases protected by an exact $\mathbb{Z}_2\times \mathbb{Z}_2$ symmetry, which is the simplest non-trivial case. We also assume an average lattice translational symmetry, even though it is not involved directly in the topological protection.  More explicitly, we consider a family of disordered Hamiltonians $\{H_\omega\}$ that each $H_\omega$ is symmetric under the following symmetry:
\begin{equation}
    S^e = \prod_i Z_{2i},\quad S^o = \prod_i Z_{2i+1}\;,
\end{equation}
and the ensemble $\{H_\omega\}$ is translationally invariant. 

For $\{H_\omega\}$ to be an average SPT, we expect that there is no spontaneous symmetry breaking, or equivalently, there is no order parameter for our ensemble. This can also be a natural consequence if we assume that $\{|\Omega_\omega\rangle\}$ is an SRE ensemble~\cite{Ma_2025}, \ie the state $|\Omega_\omega\rangle$ is short-range correlated for almost all $\omega$. Thus by our main theorem, there must be a disorder parameter for $\{|\Omega_\omega\rangle\}$, precisely, we have the following corollary:
\begin{corollary}
    Let $\{H_\omega\}$ be a disordered average-SPT ensemble with exact
$S^e$ and $S^o$ symmetry and average translational
symmetry. Then, for any sites $i<j$, there exist local operators
$O_i$ and $O_j$, supported near $i$ and $j$ respectively, such that
\begin{equation}
    \bigl\langle O_i O_j \prod_{i<2k<j} Z_{2k} \bigr\rangle_{\Omega} = O(1).
\end{equation}
Similarly, there exist local operators $O_i'$ and $O_j'$, supported near
$i$ and $j$ respectively, such that
\begin{equation}
    \bigl\langle O_i' O_j' \prod_{i<2k+1<j} Z_{2k+1} \bigr\rangle_{\Omega} = O(1).
\end{equation}

Moreover, for any fixed ensemble $\{|\Omega_\omega\rangle\}$, the symmetry parities of these disorder parameters are uniquely determined:
the operators $O_i,O_j$ need to be even under $S^e$, and the operators
$O_i',O_j'$ need to be even under $S^o$. In addition, the parity of
$O_i,O_j$ under $S^o$ and the parity of $O_i',O_j'$ under $S^e$ are equal, and is uniquely fixed by the ensemble.
\end{corollary}

We now provide the proof of this corollary. Since $\{H_\omega\}$ is an average SPT, there is no order parameter for $\{|\Omega_\omega\rangle\}$ under the symmetry $S^e$ and $S^o$, so by our main theorem, there must be a disorder parameter for $\{|\Omega_\omega\rangle\}$. This gives the existence of the operators $O_i,O_j,O_i',O_j'$. Moreover, as shown in Sec.~\ref{subsec:disorderparity}, as the disorder parameter for $S^e$, $O_i$ and $O_j$ must be even under $S^e$, and as the disorder parameter for $S^o$, $O_i'$ and $O_j'$ must be even under $S^o$. 

To show that the parity of $O_i,O_j$ under $S^o$ and the parity of $O_i',O_j'$ under $S^e$ are unique, assume the contrary, without loss of generality, that $O_i$ and $O_j$ are odd under $S^o$ but $O_i'$ and $O_j'$ are even under $S^e$. Making use of the translational symmetry, we can then choose four sites $i_1<i_2<i_3<i_4$ such that $\dist(i_k,i_l)\gg \ell$ for all $k\neq l$, with the operator $O_k$ defined near site $i_k$, and 
\begin{align}
&\mathbb{E}_\omega|\bigl\langle O_1 O_3 \prod_{i_1<2k<i_3} Z_{2k} \bigr\rangle_{\Omega_\omega}| = \delta\\
&\mathbb{E}_\omega|\bigl\langle O_2 O_4 \prod_{i_2<2k+1<i_4} Z_{2k+1} \bigr\rangle_{\Omega_\omega}| = \delta,
\end{align}
where $\delta$ is some $O(1)$ constant, $O_{1,3}$ are odd under $S^o$ and even under $S^e$, and $O_{2,4}$ are even under $S^o$ and even under $S^e$. 

We first prove that this is impossible for a gapped clean system without disorder averaging. To this end, we invoke the following lemma:
\begin{lemma}[Proposition 1, Ref.~\cite{Levin_2020}]\label{lem:Levin_prop1}
    Let $A$ be an operator that is even under $S$, has norm of at most 1 and is supported on $I_1 \cup I_2$ where $I_1, I_2$ are two intervals separated by a distance of at least $\ell$. Then,
for any symmetric state $|\Omega\rangle$ with gap $\epsilon$ in its symmetry sector, we have
\begin{equation}
\| P_{\tilde{I}_1}(\ell) P_{\tilde{I}_2}(\ell) A S_J|\Omega\rangle-|\Omega\rangle\langle\Omega| A S_J|\Omega\rangle \| \leq \mathcal{O}(\sqrt{g(\ell)})
\end{equation}
where $ P_{\tilde{I}_1}(\ell) $ is a projection operator that is even under $S$ and supported on $\tilde{I}_j \equiv\left\{i: \operatorname{dist}\left(i, I_j\right)<\ell / 2\right\}$, $J$ denotes the interval between $I_1$ and $I_2$ (going clockwise), and $g(\ell)$ is defined as in Lemma.~\ref{lem:factorization}.
\end{lemma}
Applying this lemma to operators $O_{1,2,3,4}$ and let $\ell_{\min}=\min_{k,l}[\dist(i_k,i_l)]-2\ell$, we can define
\begin{align}
     U&=P_1P_3 O_1O_3\prod_{i_1<2k\leq i_3}Z_{2k}\;,\\
    V&=P_2P_4 O_2O_4\prod_{i_2<2k+1\leq i_4}Z_{2k+1}\;,
\end{align}
where $P_k\coloneqq P_{\tilde{I}_k}(\ell_{\min})$, we have
\begin{align}
    \opnorm{U|\Omega\rangle-\delta|\Omega\rangle}& \leq O(\sqrt{g(\ell_{\min})})\;,\\
    \opnorm{V|\Omega\rangle-\delta|\Omega\rangle}& \leq O(\sqrt{g(\ell_{\min})})\;.
\end{align}
 It then follows that
\begin{align}
    \opnorm{UV|\Omega\rangle-\delta^2|\Omega\rangle}& \leq O(\sqrt{g(\ell_{\min})})\\
    \opnorm{VU|\Omega\rangle-\delta^2|\Omega\rangle}& \leq O(\sqrt{g(\ell_{\min})})\;.
\end{align}
However, by the parity of the operators $O_{1,2,3,4}$, we have $UV=-VU$, which contradicts the above inequalities for large enough $\ell_{\min}$. We thus complete the proof for the clean case.

For the dirty case, we can apply the same argument as in Sec.~\ref{subsec:exclusion}. Suppose the contrary, that there exist operators $O_{1,2,3,4}$ as described above, we can consider the events $B_{13}$ that
\begin{equation}
    |\bigl\langle O_1 O_3 \prod_{i_1<2k<i_3} Z_{2k} \bigr\rangle_{\Omega_\omega}| \ge \delta/2,
\end{equation}
and $B_{24}$ that
\begin{equation}
    |\bigl\langle O_2 O_4 \prod_{i_2<2k+1<i_4} Z_{2k+1} \bigr\rangle_{\Omega_\omega}| \ge \delta/2,
\end{equation}
with probability $\mathbb{P}(B_{13})\ge \delta/(2-\delta)$ and $\mathbb{P}(B_{24})\ge \delta/(2-\delta)$. 

{We can also take the rare region effects into account by the same argument as in Sec.~\ref{subsec: exclusion-rare}. We can take $i_{1,2,3,4}$ evenly spaced in the whole geometry, and let $R_i$ denote the event that there exists a rare region of size at least $L/12$ that covers site $i$, and let $R=\bigcup_i R_i$. Following the arguments in Sec.~\ref{subsec: exclusion-rare}, we have
\begin{equation}
    \mathbb{P}(B_{13}\cap\overline{R})=O(1)\;,\quad \mathbb{P}(B_{24}\cap\overline{R})=O(1)\;,
\end{equation}
for large enough system size $L$. Thus for any realization $\omega$ for which $B_{13}\cap \overline{R}$ occurs, there exists a $(\delta/2,L/12)$ disorder parameter on sites $i_{1,3}$ for some gapped state with the same parity as $O_{1,3}$. Similarly, for any realization $\omega$ for which $B_{24}\cap \overline{R}$ occurs, there exists a $(\delta/2,L/12)$ disorder parameter on sites $i_{2,4}$ for some gapped state  with the same parity as $O_{2,4}$. These two gapped states can be connected by a quasi-adiabatic evolution without changing the endpoint parity (see Appendix.~\ref{app:quasi-adiabatic}). By the same argument as in the clean case, this is not possible for sufficiently large system size, a contradiction. Therefore, the parity of $O_i,O_j$ under $S^o$ and the parity of $O_i',O_j'$ under $S^e$ are unique, which completes the proof of this corollary. 
}

\subsection{Lieb-Schultz-Mattis-type constraint for disordered ensembles}\label{subsec:disorderLSM}

The Lieb-Schultz-Mattis (LSM) theorem~\cite{Lieb_1961} and its generalizations~\cite{LSM_oshikawa,LSM_HigherD_Hastings,liu2025entanglement, liu2025twisted,yi2025lovasz,Gioia2021} provide powerful constraints on the low-energy physics of quantum many-body systems with certain symmetries and filling constraints. In this section, we show that our main theorem provides a powerful tool to derive LSM-type constraints for disordered ensembles~\cite{PhysRevX.13.031016,YouOshikawa,panahi2026quantumcriticalitystrongrandomness,Kimchi_2018,Xu_2025,liu2026average}. 

The key idea is to show that under certain symmetries, the system cannot have a disorder parameter, and thus by our main theorem, if the disorder ensemble is nearly gapped, it must have an order parameter, which implies spontaneous symmetry breaking. 

For simplicity, we focus on the case of 1D spin chains with an exact $\mathbb{Z}_2\times \mathbb{Z}_2$ symmetry and average translational symmetry. More explicitly, we consider a family of disordered Hamiltonians $\{H_\omega\}$ defined on a ring geometry with even number of qubits, that each $H_\omega$ is symmetric under the following symmetry:
\begin{equation}
    S^x= \prod_i X_{i},\quad S^z = \prod_i Z_{i}\;,
\end{equation}
and the ensemble $\{H_\omega\}$ is translationally invariant. Note that on a single site, the representation of the symmetry group $\mathbb{Z}_2\times \mathbb{Z}_2$ is projective ($X_iZ_i=-Z_iX_i$), so we expect that there is a LSM-type constraint for this system. More explicitly, we have the following corollary:
\begin{corollary}
    Let $\{H_\omega\}$ be a nearly gapped ensemble with exact $S^x$ and $S^z$ symmetry and average translational symmetry. Then $\{H_\omega\}$ must have an $O(1)$ order parameter for either $S^x$ or $S^z$.
\end{corollary}

To show this, we just need to show that $\{H_\omega\}$ cannot have an $O(1)$ disorder parameter for both $S^x$ and $S^z$. Assume the contrary, that $\{H_\omega\}$ has an $O(1)$ disorder parameter for both $S^x$ and $S^z$, we can then choose four sites $i_1<i_2<i_3<i_4$ such that $\dist(i_k,i_l)\gg \ell$ for all $k\neq l$, with the operator $O_k$ defined near site $i_k$, and
\begin{align}
&\mathbb{E}_\omega|\bigl\langle O_1 O_3 \prod_{i_1\leq k<i_3} X_{k} \bigr\rangle_{\Omega_\omega}| = \delta,\\
&\mathbb{E}_\omega|\bigl\langle O_2 O_4 \prod_{i_2\leq k<i_4} Z_{k} \bigr\rangle_{\Omega_\omega}| = \delta,
\end{align} 
where $\delta$ is some $O(1)$ constant, $O_{1,3}$ are even under $S^x$, and $O_{2,4}$ are even under $S^z$. Without loss of generality, we can also assume that $i_{1,2,3,4}$ are all even sites. By the average translational symmetry of the disorder ensemble, we then also have
\begin{equation}
\mathbb{E}_\omega|\bigl\langle  \hat{T} \left(O_2 O_4 \prod_{i_2\leq k<i_4} Z_{k}  \right)\hat{T}^\dagger \bigr\rangle_{\Omega_\omega}| = \delta,
\end{equation}
thus $O_2'=\hat{T}O_2\hat{T}^\dagger Z_{i_2}$ and $O_4'=\hat{T}O_4\hat{T}^\dagger Z_{i_4}$ also form a disorder parameter for $S^z$. Without loss of generality, we can assume that $O_{2,4}$ have a definite parity under $S^x$, and thus $O_2'$ and $O_4'$ have the opposite parity under $S^x$ because of the additionally dressed $Z$ operators. 

We can now apply a similar argument as in Sec.~\ref{subsec:aspt} to show that this is not possible for a sufficiently large system size. We start with the clean case. 

To this end, we apply Lemma.~\ref{lem:Levin_prop1} to operators $O_{1,2,3,4}$ and let $\ell_{\min}=\min_{k,l}[\dist(i_k,i_l)]-2\ell$, we can define
\begin{align}
     U&=P_1P_3 O_1O_3\prod_{i_1< k\leq i_3}X_{k}\;,\\
    V&=P_2P_4 O_2O_4\prod_{i_2< k\leq i_4}Z_{k}\;,\\
    V'&=P_2P_4 O_2'O_4'\prod_{i_2< k\leq i_4}Z_{k}\;,
\end{align}
where $P_k\coloneqq P_{\tilde{I}_k}(\ell_{\min})$. By Lemma.~\ref{lem:Levin_prop1}, we have
\begin{align}
    \opnorm{U|\Omega\rangle-\delta|\Omega\rangle}& \leq O(\sqrt{g(\ell_{\min})})\;,\\
    \opnorm{V|\Omega\rangle-\delta|\Omega\rangle}& \leq O(\sqrt{g(\ell_{\min})})\;,\\
    \opnorm{V'|\Omega\rangle-\delta|\Omega\rangle}& \leq O(\sqrt{g(\ell_{\min})})\;.
\end{align} 
However, since $O_2$ and $O_2'$ have opposite parity under $S^x$, we have either $UV=-VU$ or $UV'=-V'U$, which contradicts the above inequalities for large enough $\ell_{\min}$. We thus complete the proof for the clean case.

We now move on to the dirty case. We can apply the same argument as in Sec.~\ref{subsec:exclusion}. Suppose the contrary, that there exist operators $O_{1,2,3,4}$ as described above, we can consider the events $B_{13}$ that
\begin{equation}
    |\bigl\langle O_1 O_3 \prod_{i_1\leq k<i_3} X_{k} \bigr\rangle_{\Omega_\omega}| \ge \delta/2,
\end{equation}
and $B_{24}$ that
\begin{equation}
    |\bigl\langle O_2 O_4 \prod_{i_2\leq k<i_4} Z_{k} \bigr\rangle_{\Omega_\omega}| \ge \delta/2.
\end{equation}
The average translational symmetry of the disorder ensemble also implies that there is an event $B_{24}'$ that
\begin{equation}
    |\bigl\langle O_2' O_4' \prod_{i_2\leq k<i_4} Z_{k} \bigr\rangle_{\Omega_\omega}| \ge \delta/2,
\end{equation}
Similar to the analysis in Sec.~\ref{subsec:exclusion}, since the expectation values of the disorder parameters cannot be larger than 1, we can show that the probability
\begin{equation}
    \mathbb{P}(B_{13}), \mathbb{P}(B_{24}), \mathbb{P}(B_{24}')\ge \delta/(2-\delta).
\end{equation}
{We can also take the rare region effects into account by the same argument as in Sec.~\ref{subsec: exclusion-rare}. We can take $i_{1,2,3,4}$ evenly spaced in the whole geometry, and let $R_i$ denote the event that there exists a rare region of size at least $L/12$ that covers site $i$, and let $R=\bigcup_i R_i$. Following the arguments in Sec.~\ref{subsec: exclusion-rare}, we have
\begin{align}
    \mathbb{P}(B_{13}&\cap\overline{R})=O(1)\;,\\
    \mathbb{P}(B_{24}&\cap\overline{R})=O(1)\;,\\
    \mathbb{P}(B_{24}'&\cap\overline{R})=O(1)\;,
\end{align}
for large enough system size $L$. Thus there exists three gapped states with, respectively, a $(\delta/2,L/12)$ disorder parameter on sites $i_{1,3}$ with the same parity as $O_{1,3}$, a $(\delta/2,L/12)$ disorder parameter on sites $i_{2,4}$ with the same parity as $O_{2,4}$, a $(\delta/2,L/12)$ disorder parameter on sites $i_{2,4}$ with the same parity as $O'_{2,4}$.

These three gapped states can be connected by a quasi-adiabatic evolution without changing the endpoint parity (see Appendix.~\ref{app:quasi-adiabatic}). By the same argument as in the clean case, this is not possible for sufficiently large system size, a contradiction.

Therefore, $\{H_\omega\}$ cannot have an $O(1)$ disorder parameter for both $S^x$ and $S^z$. By our main theorem, if $\{H_\omega\}$ is nearly gapped, it must have an $O(1)$ order parameter for either $S^x$ or $S^z$, which completes the proof of this corollary.

}

\subsection{Fermion chain and intrinsic average SPT}
\label{subsec:JordanWigner}

Our results on the Ising spin chain can be readily translated to fermion chains through the Jordan-Wigner transformation. {In this section, we demonstrate that the trade-off between order and disorder parameters in the spin chain translates to a trade-off between disorder parameters with bosonic or fermionic endpoint operators in fermionic systems. We then apply this result to understand a class of ``intrinsically disordered'' topological phase proposed in Ref.~\cite{Ma_2025}.}

In the following, we denote the $\mathbb{Z}_2$ symmetry as generated by
\begin{equation}
    S=\prod_iX_i.
\end{equation}
The Jordan-Wigner transform maps the Ising spin chain to a fermion chain, with two Majorana fermions $\gamma_{2i-1},\gamma_{2i}$ for every unit cell $i$:
\begin{eqnarray}
    X_i&=&i\gamma_{2i-1}\gamma_{2i}=F_i, \nonumber \\
    Z_i&=&(\prod_{j<i}F_j)\gamma_{2i-1},
\end{eqnarray}
where $F_i=\pm1$ is the fermion parity on site $i$. The $\gamma$ fermions couple to a dynamical $\mathbb{Z}_2$ gauge field -- in $1d$ the only gauge-invariant degree of freedom of a $\mathbb{Z}_2$ gauge field is the gauge flux (holonomy) along the entire closed $1d$ ring, which sets the spatial boundary condition for the fermions. In the Jordan-Wigner transform, the $\mathbb{Z}_2$ gauge flux is given by the total Ising charge $S$. The significance of having exact Ising symmetry is that the $\mathbb{Z}_2$ gauge field is effectively non-dynamical once we fix the total $S$ charge, and it suffices to consider a genuine fermion chain. This allows us to translate our results on Ising chain to the fermion chain.

First, we consider a disorder operator in the Ising chain. The undressed disorder operator maps to
\begin{equation}
    S_{12}:=\prod_{i_1\leq k\leq i_2}X_k = \prod_{i_1\leq k\leq i_2}F_i:=F_{12},
\end{equation}
which is a disorder operator in fermion parity. For a dressed disorder operator $O_{1}O_2S_{12}$, if $O_{1,2}$ are $\mathbb{Z}_2$-even, we have
\begin{equation}
\label{eq:bosondisorder}
    O_{1}O_2S_{12}=\tilde{O}_1\tilde{O}_2F_{12},
\end{equation}
where $\tilde{O}_{1,2}$ are some bosonic local operators in the fermion chain. If $O_{1,2}$ are $\mathbb{Z}_2$-odd, then in the fermion picture an additional flux-changing operator has to be inserted, leading to
\begin{equation}
    O_1O_2S_{12}=f_1f_2,
\end{equation}
where $f$ are fermion operators. For the simplest choice of $O_i=Z_i$, the $f_i$ operators are nothing but the $\gamma$ fermions. More generally they can take forms like $\gamma_a\gamma_b\gamma_c$. Combining the two results, the order parameters $O_1O_2$ ($O_{1,2}$ being $\mathbb{Z}_2$ odd) in the Ising chain maps to
\begin{equation}
\label{eq:fermiondisorder}
    O_1O_2=f_1f_2 F_{12},
\end{equation}
where again $f_{1,2}$ are fermion operators. In a fermion chain, the above disorder operator dressed with fermion operators is used to detect nontrivial topological superconductivity in the Kitaev phase. 


For a nearly gapped ensemble in a fermion chain, we immediately conclude from Theorem~\ref{thm:tradeoff}:
\begin{enumerate}
    \item There always exists some disorder parameter
    \begin{equation}
        \lim_{|i_1-i_2|\to\infty}\mathbb{E}_{\omega}|\langle O_1 O_2\prod_{i_1\leq k\leq i_2}F_k\rangle_{\Omega_\omega}|= O(1).
    \end{equation}
    The dressing local operators $O_{1,2}$ are either always bosonic or always fermionic. The former corresponds to a trivial phase while the latter corresponds to a Kitaev chain.
    \item The fermion two-point correlator must decay:
    \begin{equation}
        \lim_{|i-j|\to\infty}\mathbb{E}_{\omega}\langle f_i f_j\rangle_{\Omega_\omega}= 0.
    \end{equation}
\end{enumerate}

In Ref.~\cite{Ma_2025} a simple example of \textit{intrinsically disordered} topological phase was proposed in a fermion chain. The proposed phase is essentially a disordered Kitaev chain, but with an additional average $\mathbb{Z}_4$ symmetry $c_i\to ic_i$ (where $c_i=\gamma_{2i-1}+i\gamma_{2i}$ is the complex fermion on each site). In fact the average symmetry can even be enlarged to a full $U(1)$: $c_i\to e^{i\alpha}c_i$. In the Ising chain picture this $U(1)$ is generated by charge $Q=\sum_i(1-X_i)/2$. Having this average symmetry unbroken requires the vanishing of the first-moment (not Edwards-Anderson) order parameter for any local operator $O_i$ that is $\mathbb{Z}_2$-odd (hence also charged under $U(1)$):
\begin{equation}
    \lim_{|i-j|\to\infty}\mathbb{E}_\omega\langle O_{i} O^{\dagger}_j\rangle_{\Omega_{\omega}}=0.
\end{equation}
Translating this to the fermion picture means that the first-moment fermion-dressed disorder parameter must vanish:
\begin{equation}
    \lim_{|i-j|\to\infty}\mathbb{E}_\omega\langle f_{i} f^{\dagger}_j\prod_{i\leq k\leq j}F_k\rangle_{\Omega_{\omega}}=0.
\end{equation}
This is why in a clean system with a $U(1)$ symmetry, the Kitaev phase cannot be realized as a gapped state. But in a dirty system, our results indicates that we should instead examine the Edwards-Anderson disorder parameter. In this light, a Kitaev phase with
\begin{equation}
    \lim_{|i-j|\to\infty}\mathbb{E}_\omega|\langle f_{i} f^{\dagger}_j\prod_{i\leq k\leq j}F_k\rangle_{\Omega_{\omega}}|= O(1),
\end{equation}
is perfectly compatible with the average (but not exact) $U(1)$ symmetry. This is why the phase is considered ``intrinsically disordered''. 

\section{Summary and Outlook}

In this work, we defined notions of order and disorder parameters for quenched disordered spin chains and established a fundamental trade-off between them. Specifically, we showed that: (i) any gapped ensemble must exhibit either an order parameter or a disorder parameter; (ii) the two cannot coexist; and (iii) any disorder parameter must be even under the symmetry. These results extend naturally to nearly gapped ensembles in the presence of rare-region effects.

Order and disorder parameters play a central role in the understanding of phases of matter without disorder averaging. Our results provide a firm foundation for their role in quenched disordered systems as well. Rare-region effects have long been difficult to treat in the study of disordered systems. Based on the LPPL principle, our work gives a systematic understanding of how such effects influence order and disorder parameters, which may offer broader insight into the role of rare regions in disordered phases.

Our results also have concrete applications to disordered systems. First, we prove the existence of well-defined string order parameters for SPT phases in dirty systems. Second, our theorem yields a Lieb-Schultz-Mattis-type constraint for disordered ensembles with anomalous symmetry. Third, our results can be extended to fermion systems and provide a more rigorous understanding of certain ``intrinsically disordered'' topological phases. 

More broadly, this work highlights the importance of order and disorder parameters in the study of quenched disordered systems. Promising directions for future work include generalizing these ideas beyond Ising symmetries and extending them to higher dimensions. It would also be interesting to formulate analogous notions of order and disorder parameters for mixed-state phases, such as decohered average SPT phases.

\textit{Note added}: for complementary perspective on order-disorder relation in the context of mixed quantum states, developed in parallel to ours by Divi, Lessa and Wang, see Ref.~\cite{DiviLessaWang2026}.

\begin{acknowledgments}
We thank Anton Burkov and Ruizhi Liu for valuable discussions. Research at Perimeter Institute is supported in part by the Government of Canada through the Department of Innovation, Science and Economic Development and by the Province of Ontario through the Ministry of Colleges and Universities. JY is also supported by the Natural Sciences and Engineering Research Council (NSERC) of Canada.
\end{acknowledgments}

\appendix
\section{$\varepsilon$-net for local operators}\label{app:epsilon_net}
Given a set $K$ equipped with a notion of distance, an \emph{$\varepsilon$-net} for $K$ is a finite subset $\mathcal{N}\subseteq K$ such that every element of $K$ is within distance $\varepsilon$ of some element of $\mathcal{N}$. The minimal size of such a net is the \emph{covering number} $N(K,\varepsilon)$. A standard result~\cite{papaspiliopoulos2020high} states that for a Euclidean ball $B^n(r)$ of radius $r$ in $\mathbb{R}^n$,
\begin{equation}
    N\!\left(B^n(r),\varepsilon\right) \;\le\; \Bigl(1+\frac{2r}{\varepsilon}\Bigr)^{n}\,.
\end{equation}

In Sec.~\ref{subsec: events} and Sec.~\ref{sec:Rareregion}, we need to take a union bound over all possible choices of the order and disorder parameter operators. For this to be valid with a finite collection, we need a finite representative for the order and disorder parameters, which requires the notion of a $\varepsilon$-net for local operators. We therefore introduce the following definition.
\begin{definition}[$\varepsilon$-net for $\ell$-local operators]\label{def:epsilon_net}
Let $\mathcal{O}_\ell$ denote the set of $\ell$-local operators on the qubit chain with $\opnorm{O}\le 1$, with the distance metric given by the operator norm. A finite subset $\mathcal{N}_\ell(\varepsilon)\subseteq\mathcal{O}_\ell$ is an \emph{$\varepsilon$-net} for $\mathcal{O}_\ell$ if for every $O\in\mathcal{O}_\ell$ there exists $O'\in\mathcal{N}_\ell(\varepsilon)$ such that $\opnorm{O-O'}\le\varepsilon$.
\end{definition}

We now estimate the covering number $|\mathcal{N}_\ell(\varepsilon)|$. Let the local Hilbert space on an interval of length $\ell$ have dimension $d = 2^{\ell}$. Every operator $O\in\mathcal{O}_\ell$ can be expanded in the orthonormal Pauli basis $\{P_k\}$ as $O = \sum_k c_k P_k$, with Frobenius norm
\begin{equation}
    \opnorm{O}_{\mathrm{F}} \coloneqq \sqrt{\tr(O^\dagger O)} = \sqrt{\sum_k |c_k|^2}\,.
\end{equation}
Collecting the real and imaginary parts of the coefficients identifies every operator $O\in\mathcal{O}_\ell$ with a point of $\mathbb{R}^{2d^2}$. The Frobenius and operator norms satisfy
\begin{equation}
    \opnorm{O} \le \opnorm{O}_{\mathrm{F}} \le \sqrt{d}\,\opnorm{O}\,,
\end{equation}
so every $O\in\mathcal{O}_\ell$ corresponds to a point in the Euclidean ball of radius $\sqrt{d}$ in $\mathbb{R}^{2d^2}$, and $\opnorm{O-O'}_{\mathrm{F}}\le\varepsilon$ implies $\opnorm{O-O'}\le\varepsilon$. Hence an $\varepsilon$-net for the Euclidean ball induces an $\varepsilon$-net for $\mathcal{O}_\ell$, and applying the covering-number estimate above with $r=\sqrt{d}$ and $n=2d^2$ gives an estimation for the cover number $N_\ell(\varepsilon)\coloneqq|\mathcal{N}_\ell(\varepsilon)|$.
\begin{equation}
    N_\ell(\varepsilon)\le \Bigl(1+\frac{2\sqrt{d}}{\varepsilon}\Bigr)^{2d^2}
    =\Bigl(1+\frac{2^{\ell/2+1}}{\varepsilon}\Bigr)^{2^{2\ell+1}}\,.
\end{equation}
Although this grows rapidly with $\ell$, it does not depend on the system size, and since $\ell$ is an $O(1)$ constant independent of system size, the net size remains bounded by a system-size-independent constant.
\\
\section{Quasi-adiabatic evolution for gapped symmetry sector}\label{app:quasi-adiabatic}
In this appendix we review the quasi-adiabatic evolution~\cite{Hastings_2005,Bravyi_2010,Bachmann_2011} and its implication in this work.
\subsection{Construction}
For two local Hamiltonians $H_0$ and $H_1$ smoothly connected by a path $H_s$ along which the gap remains open, the ground-state subspaces of the two systems can be connected by a quasi-adiabatic evolution defined as follows:
\begin{equation}\label{eq:adiabatic2}
    U \coloneqq \mathcal{S} \exp \left\{i \int_0^1 \mathrm{~d} s\mathcal{D}_{s}\right\}
\end{equation}
where $\mathcal{S}$ denotes an $s$-ordered exponential, and the quasi-adiabatic continuation operator $\mathcal{D}_s$ is given by
\begin{equation}\label{eq:adiabatic1}
    \mathcal{D}_s \coloneqq i \int \mathrm{~d} t F(t) \exp \left(i H_s t\right)\left( \partial_sH_s\right) \exp \left(-i H_s t\right)
\end{equation}
with $F(t)$ being a super-polynomially decaying function.

This quasi-adiabatic construction extends to the case where the spectral gap is required only within the parity-even symmetry sector, rather than across the full Hilbert space. Indeed, since $\partial_s H_s$ is parity-even, it cannot connect states of opposite parity, so the sector-restricted gap suffices to control the evolution.

Although the quasi-adiabatic evolution is generated by an almost-local Hamiltonian, it can be well approximated by a quantum circuit of $O(\log L)$ depth~\cite{Haah_2021,yi2025universaldecayconditionalmutual}. For the purpose of estimating how local operators transform under the evolution, it was also shown~\cite{yi2025universaldecayconditionalmutual} that the evolution can be approximated by a finite-depth quantum circuit, with errors controlled by the locality of the gates. In particular, if $H_s$ is a sum of local terms each commuting with $S$, the approximating circuit can be chosen so that every gate commutes with $S$; symmetry is thus preserved gate-by-gate, not merely in for the whole circuit. In what follows we exploit this and simply treat $U$ as a finite-depth, $S$-symmetric quantum circuit.

\subsection{Quasi-adiabatic evolution for disorder parameters}

We now show that the parity of the endpoint operators of a disorder parameter is invariant under quasi-adiabatic evolution. This is used in Sec.~\ref{subsec:disorderparity} to establish that any disorder parameter must be parity-even, and underlies our applications to SPT phases and the disorder-averaged LSM theorem in Sec.~\ref{sec: Application}.

Recall that a disorder parameter involves a symmetry string $S_\Gamma\coloneqq\prod_{j\in\Gamma}S_j$ over a contiguous interval $\Gamma$, together with endpoint operators $O_L$ and $O_R$ localized near the left and right endpoints of $\Gamma$, respectively. As argued above, the quasi-adiabatic evolution $U$ can be approximated by a finite-depth quantum circuit whose gates all commute with $S$. We factor this circuit as $U = U_R U_0 U_L$, where $U_L$ (resp.\ $U_R$) consists of all gates in the lightcone of the left (resp.\ right) endpoint of $\Gamma$, and $U_0$ contains the remaining bulk gates. Since the bulk gates commute with $S_\Gamma$, they pass through the string unchanged, and the dressed disorder operator decomposes as
\begin{equation}
    U^\dagger\!\left(O_L S_\Gamma O_R\right)U = \widetilde{O}_L\, S_\Gamma\, \widetilde{O}_R,
\end{equation}
where the dressed endpoint operators are
\begin{align}
    \widetilde{O}_L &= U_L^\dagger O_L U_L \cdot U_L^\dagger S_\Gamma U_L S_\Gamma^{-1}, \\
    \widetilde{O}_R &= S_\Gamma^{-1} U_R^\dagger S_\Gamma U_R \cdot U_R^\dagger O_R U_R.
\end{align}
Here the extra factors $U_L^\dagger S_\Gamma U_L S_\Gamma^{-1}$ and $S_\Gamma^{-1}U_R^\dagger S_\Gamma U_R$ results from the action of $U_L$ and $U_R$ in the undressed disorder operator.

It remains to check that the parity is preserved. Suppose $S O_L S^{-1} = p_L O_L$. Since both $U_L$ and $S_\Gamma$ commute with $S$, we have
\begin{equation}
    \begin{split}
            S\widetilde{O}_LS^{-1}
    &= S(U_L^\dagger O_L U_L)S^{-1} \cdot S(U_L^\dagger S_\Gamma U_L S_\Gamma^{-1})S^{-1}\\
    &= p_L\,\widetilde{O}_L,
    \end{split}
\end{equation}
so $\widetilde{O}_L$ carries the same parity as $O_L$. An identical argument applies to $\widetilde{O}_R$. The same reasoning extends to systems with multiple $\mathbb{Z}_2$ symmetries, as long as the bare disorder operator and all gates in the quasi-adiabatic evolution commute with all symmetries. This is indeed the case for the applications in Sec.~\ref{sec: Application}.

\bibliography{lib}
\end{document}